\documentclass[a4paper]{article}

\pdfoutput=1

\usepackage[T1]{fontenc}
\usepackage[utf8]{inputenc}
\usepackage[margin=28mm]{geometry}
\usepackage{csquotes}
\usepackage[english]{babel}
\usepackage[style=numeric-comp,sorting=none,giveninits]{biblatex}
\usepackage[font=footnotesize]{caption}
\usepackage{amsmath}
\usepackage{amssymb}
\usepackage{amsthm}
\usepackage[affil-it]{authblk}
\usepackage{booktabs}
\usepackage{braket}
\usepackage{enumitem}
\usepackage{newfloat}
\usepackage[binary-units]{siunitx}
\usepackage{subcaption}
\usepackage{thmtools, thm-restate}
\usepackage[table]{xcolor}
\usepackage{tikz}
\usepackage{tikzscale}
\usepackage{graphicx}
\usepackage[colorlinks]{hyperref}
\usepackage[capitalise]{cleveref}

\usetikzlibrary{matrix, positioning}

\newtheorem{theorem}{Theorem}
\newtheorem{proposition}[theorem]{Proposition}
\newtheorem{lemma}[theorem]{Lemma}

\theoremstyle{definition}
\newtheorem{definition}[theorem]{Definition}
\theoremstyle{remark}
\newtheorem*{remark}{Remark}

\Crefname{rule}{Rule}{Rules}

\newcommand{\bigO}{O}
\newcommand{\knownoperator}[1]{\hat{#1}}
\newcommand{\notqubit}[1]{\overline{#1}}
\newcommand{\rectdim}[2]{$#1 \times #2$}
\DeclareMathOperator{\poly}{poly}

\DeclareFloatingEnvironment{protocol}
\Crefname{protocol}{\protocolname}{\protocolname s}

\begin{document}

\title{Practical parallel self-testing of Bell states \\ via magic rectangles}

\author[ \hspace{-1ex}]{Sean A. Adamson\thanks{\href{mailto:sean.adamson@ed.ac.uk}{\texttt{sean.adamson@ed.ac.uk}}}}
\author[ \hspace{-1ex}]{Petros Wallden\thanks{\href{mailto:petros.wallden@ed.ac.uk}{\texttt{petros.wallden@ed.ac.uk}}}}
\affil[ \hspace{-1ex}]{School of Informatics, University of Edinburgh, \protect\\ 10 Crichton Street, Edinburgh EH8 9AB, United Kingdom}

\date{}

\maketitle

\begin{abstract}
Self-testing is a method to verify that one has a particular quantum state from purely classical statistics. For practical applications, such as device-independent delegated verifiable quantum computation, it is crucial that one self-tests multiple Bell states in parallel while keeping the quantum capabilities required of one side to a minimum. In this work, we use the $3 \times n$ magic rectangle games (generalizations of the magic square game) to obtain a self-test for $n$ Bell states where the one side needs only to measure single-qubit Pauli observables. The protocol requires small input sizes [constant for Alice and $O(\log n)$ bits for Bob] and is robust with robustness $O(n^{5/2} \sqrt{\varepsilon})$, where $\varepsilon$ is the closeness of the ideal (perfect) correlations to those observed. To achieve the desired self-test, we introduce a one-side-local quantum strategy for the magic square game that wins with certainty, we generalize this strategy to the family of $3 \times n$ magic rectangle games, and we supplement these nonlocal games with extra check rounds (of single and pairs of observables).
\end{abstract}

\section{Introduction}

One of the most profound properties of quantum theory---one that defies our classical intuition---is that it exhibits nonlocality \cite{bell1964einstein}.
This distinct characteristic enables us to deduce that certain results we gather in an experimental setting cannot be explained with classical notions, and that there is necessarily some underlying quantumness at work.
Even more interestingly, nonlocality makes it possible to deduce the exact quantum state of a real experimental system based on purely classical statistics.
This property is known as self-testing.
Beyond the foundational importance of being able to verify the quantum state of a totally untrusted black-box experimental setup, self-testing has many practical uses due to the higher levels of security it is able to offer.
While the standard notions of nonlocality lead to device-independent cryptography (see, for example, \cite{colbeck2011private,vazirani2014fully}), self-testing enables such applications as device-independent secure delegated (verifiable) quantum computation \cite{mckague2016interactive,hajduvsek2015device,gheorghiu2015robustness,gheorghiu2017rigidity} among other device-independent protocols that involve quantum \emph{computation}.
The crucial point is that, to enable device-independent quantum computation, one needs to test the quantum state itself (that is one must perform self-testing); simple observation of nonlocal correlations does not suffice.

Delegated verifiable blind quantum computation \cite{fitzsimons2017unconditionally,gheorghiu2019verification} is arguably one of the most important applications of self-testing.
Here, a client wishes to delegate some computation to a server (which has a quantum computer) such that the privacy of their input/output and computation is preserved, and in a way that allows the client to verify the validity of the answer that they receive.
This is a setting with increasing practical relevancy, since quantum hardware companies already offer their services in the cloud.
Protecting the privacy of client data and giving reassurances that the computation was performed as desired are crucial to making this model work.
In this setting, since one side (the server) has access to a universal quantum computer, having extra quantum operations being performed on this side as part of a self-test comes with almost no further practical limitations.
On the other hand, the client is assumed to have minimal quantum capabilities.
Moreover, the client and server should self-test multiple maximally entangled Bell states in parallel.
This is required in order to perform any interesting quantum computation (otherwise the client could simply perform the computation classically on their side).
It follows that any natural self-test for such an application will have minimal experimental requirements on one side (the client) while also being required to test for many Bell states in parallel.
This is precisely the nature of the self-test we obtain in this work.

A further observation is that self-tests of quantum states typically arise as the observation of an optimal quantum strategy for a certain nonlocal game.
Conversely, exploring how different nonlocal games that appear elsewhere in the literature can be used for self-testing and what (if any) advantages these offer over other self-tests is, in its own right, an interesting endeavor.
Here we examine the recently introduced generalization of the magic square game to rectangular dimensions \cite{adamson2020quantum}, and we obtain a family of self-tests that compare favorably to other self-tests.

It is worth mentioning that one can compare self-tests with respect to a number of different figures of merit, with the importance of each depending on the application for which one wishes to use the self-test.
We consider several different qualities, and in \cref{sec:discussion} analyze what our proposed self-tests achieve and how they compare to other works.
The first is the experimental complexity required by our self-test.
This depends on the honest strategy and determines the quantum devices and resources required by each party.
The second is that of communication complexity (required input and output sizes for the parties involved).
Its most important ingredient is that of input question size, as this determines the amount of randomness that must be consumed per round of interaction of the protocol.
This can also play an important role in other aspects, e.g., in how much randomness can be generated in possible applications to private randomness expansion.
Finally, the third figure of merit that self-tests can be compared upon is their robustness, i.e., how close to the ideal behavior the observed correlations need to be in order to ensure that the tested quantum state is sufficiently close to the desired reference state.
Given that experiments have intrinsic imperfections and correlations cannot be perfectly saturated in a real setting, achieving good robustness is crucial for practical uses of self-testing.
While many self-testing protocols are designed to perform well with respect to few particular figures of merit, it is key for the application at hand that a protocol achieves appropriate levels of performance simultaneously across all relevant areas.
This is a major consideration in the self-test we present here.

\subsection{Our contributions}

We aim to obtain an improved self-test of multiple Bell states (with respect to different figures of merit).
The nonlocal games at the core of our approach belong to the set of magic rectangle games.
Our contributions may be summarized specifically as follows:
\begin{itemize}
    \item We provide a quantum strategy to win the magic square game with certainty.
    This strategy involves three Bell pairs and, importantly, one side (say Alice) need only ever make local (single-qubit) Pauli measurements.\footnote{We refer to a measurement performed on an observer's system of qubits as \emph{local} (as opposed to \emph{entangled}) or \emph{single-qubit} if it can be realized from measurements made on individual qubits independently.}
    We say this strategy has the ``one-side-local'' property.
    \item Based on this quantum strategy, we present a one-side-local self-test of three Bell states.
    This requires the introduction of some extra ``check'' rounds.
    Compared to other self-tests using the magic square game, ours requires a simpler experimental setup (one-side-local) and certifies a greater number of Bell states in parallel.
    \item We also consider the set of \rectdim{3}{n} magic rectangle games, obtaining one-side-local quantum strategies for these (again winning with certainty) involving $n$ Bell states.
    \item From these strategies, we construct a parallel self-test of $n$ Bell pairs that is one-side-local.
    This is our main result, as it offers an experimentally simpler parallel self-test that (i) has good input size scaling in the number of Bell states (constant for Alice and logarithmic for Bob), (ii) uses only perfect correlations, and (iii) is robust with robustness $\bigO \big( n^{\frac{5}{2}}\sqrt{\varepsilon} \big)$, where $\varepsilon$ is the closeness of the ideal (perfect) correlations to those observed.
    Importantly, these properties are achieved simultaneously.
\end{itemize}

\subsection{Overview of techniques}

Our two main results are self-testing protocols for three and $n$ Bell states, respectively.
Informally, to self-test a quantum state one needs to provide a local isometry that maps an untrusted state (and operators) to a reference state (and operators), which are close to the desired ones.
Our proofs proceed in five steps.
In the first step, we define a nonlocal game (along with an optimal quantum winning strategy for that game) that will form the basis of the self-test.
Importantly, the particular strategy given should involve the states that we are testing.
In the second step, we give the honest behavior for the self-test.
This fixes the experimental requirements for each side.
The honest behavior includes (i) the optimal quantum strategy for the nonlocal game given earlier; and (ii) additional ``check'' rounds, where some further correlations (that do not need to exhibit nonlocality on their own) are requested.\footnote{It is precisely these extra checks that allow us to self-test three Bell states using a nonlocal game that is normally used to self-test two Bell states (the magic square game).
This is extended later for the $n$ Bell state case with a game whose optimal winning probability could be saturated with a quantum strategy involving just a pair of Bell states.}
In the third step, we define the (untrusted) observables and specify all the correlations that are to be tested.
This is the information we have from experiment; it quantifies the proximity of the real experiment to the ideal maximum winning probabilities, and it forms the basis for obtaining the desired isometry.
In the fourth step, the above correlations are used to prove that the untrusted single-qubit Pauli operators have commutation and anticommutation relations exactly as the corresponding (trusted) Pauli operators have.
This is the hardest step, as it demonstrates that the correlations obtained from the experiment suffice to construct some untrusted operators that behave as the desired trusted operators.
The fifth and final step is simply the application of a theorem of \textcite{coladangelo2017parallel}, wherein the existence of the desired local isometry was reduced to the satisfaction of the commutation and anticommutation relations obtained in the fourth step.

\subsubsection{Self-test of three Bell states}

\paragraph{Base nonlocal game.}
We introduce a strategy for winning the magic square game with certainty (\cref{sec:one_side_local_strategy}).
This strategy has two interesting features.
Firstly, unlike the ``standard'' strategy that involves two Bell states \cite{aravind2004quantum}, this strategy involves three Bell states.
This means that any self-test based on this would result in self-testing more Bell states in parallel than using the magic square game in the standard way \cite{wu2016device}.
To succeed in the parallel self-testing of more Bell states requires some extra correlations (obtained from some ``check'' rounds) to prevent dishonest players from simply following the standard magic square strategy using only two Bell states.
The second feature is that this strategy can be realized with Pauli measurements (as in the standard magic square strategy) but with one of the players (say Alice) needing only to perform local (single-qubit) measurements.
In the usual magic square strategy, both parties must measure in entangled bases (see \cref{sec:magic_square_game}).
This implies that a self-test based on this strategy would be simpler to execute experimentally and, importantly, impose fewer quantum-technological requirements on Alice's side---something of immediate interest for major applications of self-testing.

\paragraph{Honest run.}
Alice plays the one-side-local magic square strategy (see \cref{sec:one_side_local_strategy}), with the difference being that she measures locally each of her three qubits and returns these as her answer, allowing the product of pairs to be checked by a referee.
Bob has two types of rounds: game rounds, where he plays the modified magic square game by measuring pairs of qubits in the $\knownoperator{X} \otimes \knownoperator{X}$, $\knownoperator{Y} \otimes \knownoperator{Y}$, and $\knownoperator{Z} \otimes \knownoperator{Z}$ bases simultaneously, and ``check'' rounds, where he measures his three qubits locally.

\paragraph{Untrusted observables and correlations.}
Alice has only untrusted local Pauli observables, while Bob has different untrusted observables in game and check rounds.
Interestingly, Bob's observables in the check rounds are the ones used for the isometry, while the observables of game rounds are used to enforce the suitable commutation and anticommutation relations on Alice's side.
The correlations observed are those required for the magic square game along with the (perfect) Einstein--Podolsky--Rosen (EPR) correlations in check rounds.

\paragraph{Commutation and anticommutation.}
The main theorem for this case (\cref{thm:commutation_summary} of \cref{sec:commutations}) is stated informally here.
\begin{theorem}[Informal \cref{thm:commutation_summary}]
    The game-round observables of Alice and the check-round observables of Bob obey standard commutation and anticommutation relations up to $\bigO(\sqrt{\varepsilon})$, where $\varepsilon$ is the distance of the observed correlations from the ideal ones.
    The observables commute when acting on different qubits; commute when they are of the same type and act on the same qubit; and anticommute when they act on the same qubit and are conjugate (e.g. $X$ and $Z$).
\end{theorem}

\paragraph{Isometry.}
Using the relations provided by the aforementioned theorem and following \cite{coladangelo2017parallel}, we obtain a suitable local isometry and complete the self-test.

\subsubsection{Self-test of many Bell states}

\paragraph{Base nonlocal game.}
We introduce a strategy that wins the \rectdim{3}{n} magic rectangle game with certainty using $n$ Bell states (\cref{sec:3xn_magic_game_strategy}).
Note that the \rectdim{3}{n} magic rectangle game can also be won with only two Bell states, but our strategy enables the parallel self-test of $n$ Bell states, having the same one-side-locality as our previous result.

\paragraph{Honest run.}
Alice plays the magic rectangle strategy (see \cref{sec:3xn_magic_game_strategy}) described by measuring all of her qubits in one of the three Pauli bases (all in the same basis).
Suitable products of her outcomes can be checked for consistency in the magic rectangle game by a referee.
Bob now has three round types: game rounds, \emph{local} check rounds (in which single-qubit correlations are checked), and \emph{pair} check rounds (in which correlations of pairs of qubits are checked).

\paragraph{Untrusted observables and correlations.}
Alice has only untrusted local Pauli observables, while Bob has untrusted observables for all three round types.
The local-check-round observables are used to construct the subsequent local isometry, while the other observables are used to obtain suitable commutation and anticommutation relations.

\paragraph{Commutation and anticommutation.}
The main theorem (\cref{thm:3xn_commutation_summary} of \cref{sec:3xn_commutations}) contains the same type of relations as in the case with three Bell states, where obtaining the anticommutation relations is considerably more complicated (and requires the extra set of rounds).
This is stated informally as follows:
\begin{theorem}[Informal \cref{thm:3xn_commutation_summary}]
    The game-round observables of Alice and the local-check-round observables of Bob obey standard commutation relations up to $\bigO(\sqrt{\varepsilon})$ and anticommutation relations up to $\bigO(n \sqrt{\varepsilon})$, where $\varepsilon$ is the distance of the observed correlations from the ideal ones.
    The observables commute when acting on different qubits; commute when they are of the same type and act on the same qubit; and anticommute when they act on the same qubit and are conjugate (e.g. $X$ and $Z$).
\end{theorem}

\paragraph{Isometry.}
Again following \cite{coladangelo2017parallel} and using the relations provided by \cref{thm:3xn_commutation_summary}, we recover the desired local isometry that results in a self-test of $n$ Bell states.

\subsection{Related works}

The magic square game was first introduced by \textcite{mermin1990simple,peres1990incompatible}.
\textcite{aravind2004quantum} gives a nontechnical demonstration of the Mermin--Peres magic square game.
In a previous work, we examined an extension of the magic square game to arbitrary rectangular dimensions \cite{adamson2020quantum}.
A family of these games is used as the basis for the self-test presented here.

The concept of self-testing was first introduced by \textcite{mayers1998quantum} in a cryptographic context, with the first mention of the term ``self-testing'' appearing in \cite{mayers2004self}.
\textcite{wu2016device} gave the first self-test of two maximally entangled pairs of qubits based on the magic square game, making use of the work of \textcite{mckague2016self} on self-testing in parallel.
\textcite{coladangelo2017parallel,coudron2016parallel} independently gave robust parallel self-tests of arbitrarily many Bell states based on the magic square game.
A result of \textcite{coladangelo2017parallel}, which is in turn based on results of \textcite{chao2018test}, is used in the present work (see \cref{thm:isometry}).
\textcite{natarajan2017quantum} gave the first example of a self-test for $n$ Bell states with constant robustness.
Subsequent work by the same authors achieved such a test where the number of bits of communication required is logarithmic in $n$ \cite{natarajan2018low}.
A variant of this by \textcite{natarajan2019neexp} called the ``Pauli basis test'' is presented as part of the work of \textcite{ji2021mip}.
Work in another direction is offered by \textcite{supic2021device}, who exhibit (without consideration of robustness) a constant-input-size parallel self-test for many copies of an arbitrary state given a self-test for a single copy.
On self-testing maximally-entangled states of arbitrary local dimension $d$, the results of \textcite{fu2022constant} and \textcite{manvcinska2021constant} provide robust self-tests using constant-sized questions and answers.
However, the robustness of the former is exponential in $d$ and in the latter is not constructed.
\textcite{sarkar2021self} also provide such a self-test, however, its robustness is not studied.
More details on self-testing can be found in the recent and excellent review by \textcite{supic2020self}.

\subsection{Organization of the paper}

In \cref{sec:preliminaries} some background on the properties of quantum states, self-testing, and the magic square and rectangles nonlocal games is given.
In \cref{sec:one_side_local_strategy} a one-side-local optimal quantum winning strategy for the magic square game is given, and in \cref{sec:selftest_3} this strategy is used as the basis of a parallel, one-side-local self-test of three Bell states.
In \cref{sec:selftest_n} a generalization of this one-side-local quantum strategy for \rectdim{3}{n} magic rectangle games is given, and the corresponding self-test for $n$ Bell states is proven.
We conclude in \cref{sec:discussion}.

\section{Preliminaries}
\label{sec:preliminaries}

\subsection{States and measurements}

We let registers of observers Alice and Bob be labeled by the letters $A$ and $B$ respectively.
A local Hilbert space of Alice will be denoted $\mathcal{H}_{A}$, and similarly a local Hilbert space of Bob by $\mathcal{H}_{B}$.
Sometimes we will need to talk about different Hilbert spaces local to an observer's subsystem.
For this, we will use notation such as $\mathcal{H}_{A}^{\prime}$ or $\tilde{\mathcal{H}}_{A}$ to mean different Hilbert spaces on Alice's side.
The set of linear operators on $\mathcal{H}$ will be denoted $\mathcal{L}(\mathcal{H})$.
The action of an operator $Q_{A} \in \mathcal{L}(\mathcal{H}_{A})$ on a multipartite state $\ket{\Psi} \in \mathcal{H}_{A} \otimes \mathcal{H}_{B}$ will often be shortened as $Q_{A} \ket{\Psi} = Q_{A} \otimes I_{B} \ket{\Psi}$.

For the purposes of self-testing, all quantum measurements will be defined to have two possible outcomes labeled by $\pm 1$.
We will not make any other assumptions about the physical state spaces of Alice and Bob.
In particular, we will not assume their dimensions.
We take all unknown measurements to be projective on some unknown state, with observables of the form $M = M_{+} - M_{-}$ for some orthogonal projections $M_{+}$ and $M_{-}$ satisfying $M_{+} + M_{-} = I$ and $M_{+} M_{-} = M_{-} M_{+} = 0$.
With our definitions, all unknown observables are also unitary operators and satisfy the involutory property $M^2 = I$.
Such operators that are both Hermitian and unitary are also known as \emph{reflection} operators.
That an operator is not unknown (but is instead a \emph{reference} operator) will be denoted by a hat symbol, for example the Pauli $\knownoperator{X}$ observable.

The following lemma will be useful to estimate the action of unknown observables on an unknown state.
The norm $\lVert \cdot \rVert$ associated with a Hilbert space will refer to that induced by its inner product throughout:
\begin{lemma}
\label{lem:state_estimate}
    Let $\ket{\varphi}$ and $\ket{\chi}$ be normalized states belonging to the same Hilbert space and let $\varepsilon \geq 0$.
    If the real part $\Re{\braket{\varphi | \chi}} \geq 1 - \varepsilon$ then $\lVert \ket{\varphi} - \ket{\chi} \rVert \leq \sqrt{2 \varepsilon}$.
\end{lemma}
\begin{proof}
    Immediate from the definition of the induced norm $\lVert \ket{v} \rVert = \sqrt{\braket{v | v}}$.
\end{proof}
\begin{remark}
    In the ideal case of $\varepsilon = 0$, we get $\ket{\varphi} = \ket{\chi}$.
\end{remark}

We will denote by $\ket{\Phi^{+}}$ the maximally entangled Bell state shared between Alice and Bob
\begin{equation}
    \ket{\Phi^{+}}_{AB} = \frac{\ket{0}_{A} \otimes \ket{0}_{B} + \ket{1}_{A} \otimes \ket{1}_{B}}{\sqrt{2}} .
\end{equation}
In cases where Alice and Bob share multiple such states, we may label each by an additional index so that each qubit of an observer's register can be uniquely identified.
That is, we may write
\begin{equation}
    \ket{\Phi^{+}}_{AB}^{(i)} = \frac{\ket{0}_{A}^{i} \otimes \ket{0}_{B}^{i} + \ket{1}_{A}^{i} \otimes \ket{1}_{B}^{i}}{\sqrt{2}} .
\end{equation}
To denote the case of $n$ copies of such states, with one half of each being held by Alice and the other by Bob, we will adopt the notation
\begin{equation}
    \ket{\Phi^{+}}_{AB}^{\otimes n} = \bigotimes_{i=1}^{n} \ket{\Phi^{+}}_{AB}^{(i)} .
\end{equation}

\subsection{Self-testing}

Consider local measurements made on a system shared by two observers, Alice and Bob, who are unable to communicate with one another.
Self-testing is a procedure which allows the observers to deduce the quantum state they share from purely classical (and device-independent) observations.
Specifically, given a probability distribution defining the behavior of untrusted measurement devices held by Alice and Bob, it is often possible to deduce (up to some local isometry) the quantum state they share.
Moreover, one can also often deduce the local quantum measurements corresponding to different inputs and outputs for each device.

Instead of the physical unknown state being specified by a density operator $\rho$ on $\mathcal{H}_{A} \otimes \mathcal{H}_{B}$, we will work throughout with purifications $\ket{\Psi} \in \mathcal{H}_{A} \otimes \mathcal{H}_{B} \otimes \mathcal{H}_{P}$ for some purifying space $\mathcal{H}_{P}$ separate from both observers.
This is for the sake of mathematical convenience and, since all operations accessible to the observers will act trivially on this purifying space, we will often suppress it in our notation.

Let us denote a possible output of Alice upon an input $x$ by $a$.
Similarly, upon an input $y$, let $b$ represent an output of Bob.
A fixed configuration of probabilities $p(a, b \mid x, y)$ defines a behavior for the observers.
Self-testing relies on the Born rule to express such probabilities in terms of quantum correlations $p(a, b \mid x, y) = \braket{\Psi | M_{a \mid x} \otimes N_{b \mid y} | \Psi}$ for some measurements $\{ M_{a \mid x} \}_{a} \subset \mathcal{L}(\mathcal{H}_{A})$ for Alice and $\{ N_{b \mid y} \}_{b} \subset \mathcal{L}(\mathcal{H}_{B})$ for Bob.
An isometry $\Phi \colon \mathcal{H}_{A} \otimes \mathcal{H}_{B} \to \mathcal{H}_{A}^{\prime} \otimes \mathcal{H}_{B}^{\prime}$ is called \emph{local} if it can be written as $\Phi = \Phi_{A} \otimes \Phi_{B}$ for some isometries $\Phi_{D} \colon \mathcal{H}_{D} \to \mathcal{H}_{D}^{\prime}$, where $D$ stands for either $A$ or $B$.
We are now ready to state what it means to self-test quantum states.

\begin{definition}[Self-testing of states]
    A behavior defined by correlations $p(a, b \mid x, y)$ is said to $\delta$-approximately \emph{self-test} the state $\ket{\Psi^{\prime}} \in \mathcal{H}_{A}^{\prime} \otimes \mathcal{H}_{B}^{\prime}$ if, for any state $\ket{\Psi} \in \mathcal{H}_{A} \otimes \mathcal{H}_{B} \otimes \mathcal{H}_{P}$ from which these correlations may arise, there exists a junk state $\ket{\xi} \in \tilde{\mathcal{H}}_{A} \otimes \tilde{\mathcal{H}}_{B} \otimes \mathcal{H}_{P}$ and isometries $\Phi_{D} \colon \mathcal{H}_{D} \to \mathcal{H}_{D}^{\prime} \otimes \tilde{\mathcal{H}}_{D}$ defining the local isometry $\Phi = \Phi_{A} \otimes \Phi_{B} \otimes I_{P}$ such that
    \begin{equation}
        \lVert \Phi \ket{\Psi} - \ket{\Psi^{\prime}} \otimes \ket{\xi} \rVert \leq \delta .
    \end{equation}
\end{definition}

The definition of self-testing given here can be extended to the case that we wish to self-test some quantum measurements in addition to a state.
\begin{definition}[Self-testing of measurements]
\label{def:self_test_measurements}
    A behavior $p(a, b \mid x, y)$ is said to $\delta$-approximately \emph{self-test} the state $\ket{\Psi^{\prime}} \in \mathcal{H}_{A}^{\prime} \otimes \mathcal{H}_{B}^{\prime}$ and measurements $\{ \knownoperator{M}_{a \mid x} \}_{a} \subset \mathcal{L}(\mathcal{H}_{A}^{\prime})$ and $\{ \knownoperator{N}_{b \mid y} \}_{b} \subset \mathcal{L}(\mathcal{H}_{B}^{\prime})$ if, for any state $\ket{\Psi} \in \mathcal{H}_{A} \otimes \mathcal{H}_{B} \otimes \mathcal{H}_{P}$ and measurements $\{ M_{a \mid x} \}_{a} \subset \mathcal{L}(\mathcal{H}_{A})$ and $\{ N_{b \mid y} \}_{b} \subset \mathcal{L}(\mathcal{H}_{B})$ from which these correlations may arise, there exists a junk state $\ket{\xi} \in \tilde{\mathcal{H}}_{A} \otimes \tilde{\mathcal{H}}_{B} \otimes \mathcal{H}_{P}$ and isometries $\Phi_{D} \colon \mathcal{H}_{D} \to \mathcal{H}_{D}^{\prime} \otimes \tilde{\mathcal{H}}_{D}$ defining the local isometry $\Phi = \Phi_{A} \otimes \Phi_{B} \otimes I_{P}$ such that
    \begin{equation}
        \big\lVert \Phi M_{a \mid x} N_{b \mid y} \ket{\Psi} - \knownoperator{M}_{a \mid x} \knownoperator{N}_{b \mid y} \ket{\Psi^{\prime}} \otimes \ket{\xi} \big\rVert \leq \delta
    \end{equation}
    for all $a$, $b$, $x$, and $y$.
\end{definition}

Since all unknown observables we will be dealing with take the form $M = M_{+} - M_{-}$ satisfying $M_{+} + M_{-} = I$, we can always write each measurement operator as $M_{\pm} = (I \pm M)/2$.
To self-test a state $\ket{\Psi^{\prime}}$ and a measurement $\{ \knownoperator{M}_{+}, \knownoperator{M}_{-} \}$ (having observable $\knownoperator{M} = \knownoperator{M}_{+} - \knownoperator{M}_{-}$ and acting nontrivially only on one side of the reference space) according to \cref{def:self_test_measurements}, it is sufficient by the linearity of isometries to instead show both
\begin{align}
    \big\lVert \Phi \ket{\Psi} - \ket{\Psi^{\prime}} \otimes \ket{\xi} \big\rVert &\leq \delta , \\
    \big\lVert \Phi M \ket{\Psi} - \knownoperator{M} \ket{\Psi^{\prime}} \otimes \ket{\xi} \big\rVert &\leq \delta .
\end{align}

The following theorem of \textcite{coladangelo2017parallel} (based closely on the work of \textcite{chao2018test}) allows us to deduce the existence of a local isometry required for the parallel self-testing of $n$ Bell states and single-qubit Pauli observables.
Rather than using a behavior of the observers directly, the theorem states sufficient conditions in terms of appropriate correlation, anticommutation, and commutation relations of unknown observables available to Alice and Bob.
Much of the current work will be dedicated to proving such relations from certain given correlations.
\begin{theorem}[{\cite[Theorem~3.5]{coladangelo2017parallel}}]
\label{thm:isometry}
    Let $\ket{\Psi} \in \mathcal{H}_{A} \otimes \mathcal{H}_{B}$, where $\mathcal{H}_{A}$ and $\mathcal{H}_{B}$ have even dimension.
    Suppose there exist balanced reflections $X_{A}^{i}, Z_{A}^{i} \in \mathcal{L}(\mathcal{H}_{A})$ and $X_{B}^{i}, Z_{B}^{i} \in \mathcal{L}(\mathcal{H}_{B})$ for $i \in \{1,\dots,n\}$ such that, for $D$ either $A$ or $B$ and for all distinct $i$ and $j$, they satisfy
    \begin{subequations}
    \begin{alignat}{2}
        \big\lVert \big( M_{A}^{i} - M_{B}^{i} \big) \ket{\Psi} \big\rVert &\leq \delta , \\
        \big\lVert \big\{ X_{D}^{i} , Z_{D}^{i} \big\} \ket{\Psi} \big\rVert &\leq \delta , \\
        \big\lVert \big[ M_{D}^{i} , N_{D}^{j} \big] \ket{\Psi} \big\rVert &\leq \delta ,
    \end{alignat}
    \end{subequations}
    where $M$ and $N$ can be either of $X$ and $Z$.
    Then, there exists a state $\ket{\xi} \in \tilde{\mathcal{H}}_{A} \otimes \tilde{\mathcal{H}}_{B}$ and a local isometry $\Phi = \Phi_{A} \otimes \Phi_{B}$, where $\Phi_{D} \colon \mathcal{H}_{D} \to \big( \mathbb{C}^2 \big)^{\otimes n} \otimes \tilde{\mathcal{H}}_{D}$, such that for all $i$
    \begin{subequations}
    \begin{align}
        \big\lVert \Phi \ket{\Psi} - \ket{\Phi^{+}}_{AB}^{\otimes n} \otimes \ket{\xi} \big\rVert
        &\in \bigO \big( n^{\frac{3}{2}} \delta \big) , \\
        \big\lVert \Phi M_{D}^{i} \ket{\Psi} - \knownoperator{M}_{D}^{i} \ket{\Phi^{+}}_{AB}^{\otimes n} \otimes \ket{\xi} \big\rVert
        &\in \bigO \big( n^{\frac{3}{2}} \delta \big) ,
    \end{align}
    \end{subequations}
    where $\knownoperator{X}_{D}^{i}$ and $\knownoperator{Z}_{D}^{i}$ are Pauli observables acting on the $i$th qubit subsystem of register $D$.
\end{theorem}

The assumptions of \cref{thm:isometry} that the unknown state spaces $\mathcal{H}_{A}$ and $\mathcal{H}_{B}$ have even dimension and that the unknown reflection operators acting on these spaces are balanced (that is their $+1$ and $-1$ eigenspaces have equal dimension) are not an issue for self-testing.
In the construction of the isometry, one can always extend the $\mathcal{H}_{D}$ by direct sum with Hilbert spaces of appropriate dimensions on which the extension of $\ket{\Psi}$ is defined to have no mass, and correspondingly extend each reflection to have eigenspaces of equal dimensions.
Thus we may freely assume these are automatically satisfied by any unknown reflections defined later as part of our self-testing proofs.

\subsection{The magic square game}
\label{sec:magic_square_game}

The Mermin--Peres magic square game \cite{aravind2004quantum} consists of two players, Alice and Bob, who are not allowed to communicate during each round of the game.
This could be achieved, for example, by ensuring a spacelike separation between the two players.
Each round consists of Alice and Bob, respectively, being assigned a row and column of an empty $3 \times 3$ table uniformly at random, which they must fill according to the rules:
\begin{enumerate}[label=S\arabic*.,ref=S\arabic*]
    \item \label[rule]{rule:plus_minus} Each filled cell must belong to the set $\{+1, -1\}$.
    \item \label[rule]{rule:alice_prod} Rows must contain an even number of negative entries (i.e., the product of Alice's entries to any assigned row must be $+1$).
    \item \label[rule]{rule:bob_prod} Columns must contain an odd number of negative entries (i.e., the product of Bob's entries to any assigned column must be $-1$).
\end{enumerate}
Neither player has knowledge of which row or column the other has been assigned, nor does either player know what values the other has entered.
The game is won if both players enter the same value into the cell shared by their row and column.
It is clear that the optimal classical strategy succeeds with probability $8/9$ only \cite{brassard2005quantum}, and may be achieved by both players agreeing to each follow a particular configuration for their entire table before the game begins.
Strikingly, if the players are allowed to share an entangled quantum state, it has been shown to be possible for them to win the magic square game with certainty \cite{mermin1990simple,peres1990incompatible}.

A possible quantum winning strategy for the magic square allows the players to share the entangled state
\begin{equation}
\label{eq:double_bell_state}
    \ket{\Phi^{+}}_{AB}^{(1)} \otimes \ket{\Phi^{+}}_{AB}^{(2)} .
\end{equation}
Depending on which row and column are assigned, the players make measurements on their respective quantum systems according to the observables given in the corresponding cells of \cref{fig:3x3_strategy}.
The outcomes of these determine the values which Alice and Bob should enter into their respective row and column to win with certainty.
Moreover, \cref{fig:3x3_strategy} shows that (unlike, say, the CHSH game) optimal strategies can be implemented by performing measurements of the two-qubit Pauli group only.

In the context of practical quantum strategies, we refer to measurements as \emph{local} in the sense that they are performed on only a single-qubit register.
It will be important for our purposes to understand that the strategy depicted here \textbf{cannot} be implemented, for either player, entirely with local measurements.
To see this for Bob, consider the measurements contained in the second column of \cref{fig:3x3_strategy}.
Upon this column being selected, Bob is required to answer with three bits, produced by a measurement performed on his subsystem.
The measurement, as given, is implemented as the simultaneous measurement of three observables (one for each bit of the answer).
While the three corresponding observables $\knownoperator{X} \otimes \knownoperator{X}$, $\knownoperator{Y} \otimes \knownoperator{Y}$, and $\knownoperator{Z} \otimes \knownoperator{Z}$ are compatible when considered over Bob's entire subsystem, he cannot generally perform the six component measurements on his two registers independently and then combine the outcomes to obtain the required three-bit answer; the six local measurements $\knownoperator{X} \otimes I$, $I \otimes \knownoperator{X}$, $\knownoperator{Y} \otimes I$, $I \otimes \knownoperator{Y}$, $\knownoperator{Z} \otimes I$, and $I \otimes \knownoperator{Z}$ do not all commute in pairs, and thus the measurement cannot be realized as the simultaneous measurement of these six local observables.
Similarly, consideration of the second row of \cref{fig:3x3_strategy} shows that the strategy for Alice also cannot be implemented by performing only local measurements.
We present in \cref{sec:one_side_local_strategy} a strategy for the magic square game that can be realized using only local measurements for one of the players, at the cost of requiring three shared Bell states.

\begin{figure}
    \centering
    \includegraphics{figures/figure1.tikz}
    \caption{
        A quantum strategy for the magic square game, in which the players share the entangled state given in \cref{eq:double_bell_state}.
        Observables $\knownoperator{X}$, $\knownoperator{Y}$, and $\knownoperator{Z}$ are the Pauli spin operators, and $I$ is the identity operator.
        Measurements of Alice correspond to a row, and those of Bob to a column.
        This strategy cannot be realized with either player performing only measurements localized to single-qubit registers.
        }
    \label{fig:3x3_strategy}
\end{figure}

\subsection{Magic rectangle games}

The magic square game can be generalized to be played on an \rectdim{m}{n} table \cite{adamson2020quantum}.
Such a \emph{magic rectangle} game corresponds to $m$ possible questions for Alice and $n$ for Bob.
To avoid trivially winning strategies, the game rules are generalized accordingly.

\begin{definition}[Magic rectangle games]
\label{def:magic_rectangle}
    An \rectdim{m}{n} game is specified by fixing some $\alpha_{1},\dots,\alpha_{m}$ and $\beta_{1},\dots,\beta_{n}$ each belonging to $\{+1, -1\}$, such that their product satisfies
    \begin{equation}
        \label{eq:prod_rule}
        \alpha_{1} \dots \alpha_{m} \cdot \beta_{1} \dots \beta_{n} = -1 .
    \end{equation}
    The rules of the given game are then:
    \begin{enumerate}[label=R\arabic*.,ref=R\arabic*]
        \item \label[rule]{rule:rect_plus_minus} Each filled cell must belong to the set $\{+1, -1\}$.
        \item \label[rule]{rule:rect_alice_prod} Upon being assigned the $i$th row, the product of Alice's entries must be $\alpha_{i}$.
        \item \label[rule]{rule:rect_bob_prod} Upon being assigned the $j$th column, the product of Bob's entries must be $\beta_{j}$.
    \end{enumerate}
    As before, the game is won if both players enter the same value into their shared cell.
\end{definition}
The \rectdim{3}{3} magic square game described in \cref{sec:magic_square_game} is simply the special case where $\alpha_{1} = \alpha_{2} = \alpha_{3} = 1$ and $\beta_{1} = \beta_{2} = \beta_{3} = -1$.
In fact, there are $2^{m+n-1}$ different specifications of \rectdim{m}{n} games allowed by \cref{eq:prod_rule}.

We will later be concerned specifically with \rectdim{3}{n} games in which entries to rows must all have positive products and entries to columns must all have negative products.
Such games are defined by $\alpha_i = 1$ and $\beta_j = -1$ for all $i$ and $j$ and must have odd $n$ due to \cref{eq:prod_rule}.
A particular class of winning strategies for these games will be used to build part of our self-test of $n$ Bell states.
In the case of these particular games, and as opposed to \cite{adamson2020quantum}, we can rephrase the definition of magic rectangles in a way that will prove more useful for our self-testing purposes.
If $(p_1, \dots, p_n) \in \{+1,-1\}^n$ is any possible output row of Alice (whose product is required to be $+1$), then there exists an assignment of $a_1, \dots, a_n \in \{+1,-1\}$ such that $p_j = \prod_{k \neq j} a_k$ for all $j$.
To see this, simply take $a_k = p_k$ for all $k$.
Conversely then, we may ask that Alice outputs some $a_1, \dots, a_n \in \{+1,-1\}$ and leave it to the game referees to check whether the appropriate products $p_j = \prod_{k \neq j} a_k$ form a winning row.
Notice in our special case of $n$ odd, such $p_j$ automatically satisfy the rule for Alice's rows $\prod_{j=1}^{n} p_j = +1$ for any assignment of the $a_k$.
We now rephrase the definition of \rectdim{3}{n} magic rectangle games in this special case.

\begin{definition}[\rectdim{3}{n} magic games]
\label{def:3xn_magic_rectangle}
    Given $n$ odd, Alice and Bob receive inputs $x \in \{1,2,3\}$ and $y \in \{1,\dots,n\}$, respectively.
    Alice outputs $n$ bits $a_1, \dots, a_n \in \{+1,-1\}$.
    Bob outputs $(b_1, b_2, b_3) \in \{+1,-1\}^3$ required to satisfy $b_1 b_2 b_3 = -1$.
    The game is won if $\prod_{k \neq y} a_k = b_x$.
\end{definition}
\begin{remark}
    While Bob's output here is column $y$ of a magic rectangle, Alice's output corresponds to filling row $x$ as $(p_1, \dots, p_n)$ where $p_j = \prod_{k \neq j} a_k$.
    The win condition is then equivalent to the familiar case when both players enter the same value into the shared cell $p_y = b_x$.
\end{remark}

\section{One-side-local magic square strategy}
\label{sec:one_side_local_strategy}

Recall that the usual quantum winning strategy for the magic square game requires some measurements of both Alice and Bob to be performed in entangled bases (see the discussion of \cref{sec:magic_square_game}).
We now propose a quantum strategy for the magic square game, also winning with certainty, which can be realized under the additional constraint that Alice may only make measurements localized to single qubits of her quantum system.
Each round begins by allowing Alice and Bob to share three Bell states
\begin{equation}
    \label{eq:triple_bell_state}
    \ket{\Psi} = \ket{\Phi^{+}}_{AB}^{(1)} \otimes \ket{\Phi^{+}}_{AB}^{(2)} \otimes \ket{\Phi^{+}}_{AB}^{(3)} .
\end{equation}
Half of each Bell state is given to Alice, and the other half to Bob.
The proposed measurement strategy is depicted in \cref{fig:local_strategy}.

\begin{figure}
    \centering
    \includegraphics{figures/figure2.tikz}
    \caption{
        The proposed magic square strategy.
        To realize any particular row, Alice is only required to measure each of her qubits locally, as the observables to be measured for any individual one of her three qubits commute within each row.
        }
    \label{fig:local_strategy}
\end{figure}

Notice in \cref{fig:local_strategy} that each row is formed out of commuting observables whose product is equal to the identity operator.
Similarly, the observables in each column commute and have a product equal to minus the identity operator.
Moreover, the eigenvalues of each observable are $+1$ and $-1$.
These facts combined show that \cref{rule:plus_minus,rule:alice_prod,rule:bob_prod} in \cref{sec:magic_square_game} are automatically satisfied by the outcomes of measuring a full row or column.
If $\knownoperator{M}_{A}$ is any observable for Alice's system contained in \cref{fig:local_strategy}, and if $\knownoperator{M}_{B}$ is the observable of the same cell for Bob's system, then it is easy to show the correlation
\begin{equation}
    \braket{\Psi | \knownoperator{M}_{A} \knownoperator{M}_{B} | \Psi} = 1 .
\end{equation}
This can be seen, for example, by writing the Bell states comprising the shared state of \cref{eq:triple_bell_state} in terms of eigenstates of the $\knownoperator{X}$, $\knownoperator{Y}$, and $\knownoperator{Z}$ operators respectively
\begin{equation}
    \ket{\Phi^{+}}
    = \frac{\ket{+} \otimes \ket{+} + \ket{-} \otimes \ket{-}}{\sqrt{2}}
    = \frac{\ket{+i} \otimes \ket{-i} + \ket{-i} \otimes \ket{+i}}{\sqrt{2}}
    = \frac{\ket{0} \otimes \ket{0} + \ket{1} \otimes \ket{1}}{\sqrt{2}} .
\end{equation}
Alice, therefore, always measures the same outcome as Bob for the shared cell (either both $+1$ or both $-1$), and so they win the game with certainty.

For any particular row assigned to Alice, it is clear from inspection of \cref{fig:local_strategy} that she need only make single-qubit measurements; for any given qubit of her system, the single-qubit observables she is required to measure with respect to that qubit of her register mutually commute within the row.
That is, it is always possible for Alice to realize the required observables by recording the measurement outcomes of a particular Pauli operator ($\knownoperator{X}$, $\knownoperator{Y}$, or $\knownoperator{Z}$ depending on the row) on each one of her three qubits.
This strategy can thus be phrased naturally for the magic square game in the sense of \cref{def:3xn_magic_rectangle} with $n = 3$.
Bob generates his outputs according to the columns of \cref{fig:local_strategy} as usual.
The $j$th output bit $a_j$ of Alice, however, results from the outcome of the single-qubit Pauli measurement $\knownoperator{X}_{A}^{j}$, $\knownoperator{Y}_{A}^{j}$, or $\knownoperator{Z}_{A}^{j}$ on Alice's $j$th qubit depending on whether the first, second, or third row was assigned, respectively.

\section{Self-test of three Bell states}
\label{sec:selftest_3}

By augmenting the correlations arising from a winning magic square strategy by certain additional correlations that ensure Alice implements her side of the strategy locally, it is possible to self-test three copies of the Bell state $\ket{\Phi^{+}}$.
These additional correlations are obtained from Bob making single-qubit Pauli measurements of his qubits in some rounds of the test, which we will call ``check'' rounds.
Rounds that are not check rounds will be called ``game'' rounds.
We now describe the structure of the self-test and specify its honest behavior.
Afterwards, we exhibit explicitly the correlations of unknown observables used in the test.
Finally, we show how these correlations can be used to prove the relevant commutation and anticommutation relations required for a self-testing proof.

\subsection{Structure and honest behavior}
\label{sec:structure_honest}

Alice receives an input $x \in \{1,2,3\}$ and Bob an input $y \in \{1,2,3\}$.
Additionally, Bob receives an input $c \in \{0,1\}$ controlling whether the round is a game or check round.
If the round is a game round ($c = 0$), then it is the goal of the players to win at the magic square game (in the sense of \cref{def:3xn_magic_rectangle}) with the row and column assigned to Alice and Bob given by $x$ and $y$, respectively.
Otherwise, if the round is a check round ($c = 1$), then the players are required to perfectly correlate certain combinations of their output bits (which will be convenient to state after our description of the honest behavior).
Notice, however, that Alice is not directly provided with the information of whether the round is to be considered a game or check round.
The protocol is summarized in \cref{prot:protocol_3}.

\begin{protocol}[htb]
    \caption{
        A protocol for certifying three Bell states.
        Strategies in which Alice uses entangled measurements are ruled out by \emph{local check} rounds.
        The protocol is phrased in terms of the parameter $n$, as it will be extended in \cref{sec:3xn_honest} in order to self-test $n$ Bell states.
    }
    \label{prot:protocol_3}
    \rule{\linewidth}{0.08em}
    Let $n = 3$ be the number of Bell states to be certified.
    In each round, a verifier chooses $c \in \{0,1\}$ and $y \in \{1,\dots,n\}$.
    The verifier sends Bob $(c,y)$ and, depending on $c$, runs one of the following subprotocols:
    \begin{enumerate}
        \setcounter{enumi}{-1}
        \item \emph{Magic game}:
        Send Alice $x \in \{1,2,3\}$.
        Alice and Bob answer with $a_{1}, \dots, a_{n}$ and $b_{1}, b_{2}, b_{3}$ in $\{+1,-1\}$ satisfying $b_{1} b_{2} b_{3} = -1$.
        Accept if and only if $\prod_{k \neq y} a_{k} = b_{x}$.
        \item \emph{Local check}:
        Send Alice $x \in \{1,3\}$.
        Alice and Bob answer with $a_{1}, \dots, a_{n}$ and $b_{1}, \dots, b_{n}$ in $\{+1,-1\}$.
        \begin{enumerate}
            \item If $x = 1$, accept if and only if $a_{y} = b_{y}$.
            \item If $x = 3$, accept if and only if $a_{j} = b_{j}$ for all $j \neq y$.
        \end{enumerate}
    \end{enumerate}\vspace*{-\baselineskip}
    \rule{\linewidth}{0.08em}
\end{protocol}

In an honest round of the experiment, the players share three Bell states, so that $\ket{\Psi} = \ket{\Phi^{+}}_{AB}^{\otimes 3}$ as in the magic square strategy of \cref{sec:one_side_local_strategy}.
Alice always performs her side of this magic square strategy, providing each of her output bits $a_j$ to the referees (as in \cref{def:3xn_magic_rectangle}) by measuring
\begin{subequations}
\label{eq:alice_honest_observables}
\begin{alignat}{2}
    & \knownoperator{X}_{A}^{j} && \quad \text{if $x = 1$,} \\
    & \knownoperator{Y}_{A}^{j} && \quad \text{if $x = 2$,} \\
    & \knownoperator{Z}_{A}^{j} && \quad \text{if $x = 3$.}
\end{alignat}
\end{subequations}

The honest behavior of Bob depends on the type of round $c$.
If $c = 0$, then Bob also performs his side of our one-side-local magic square strategy, returning outputs according to measuring the observables in column $y$ of \cref{fig:local_strategy} so that the magic square game is won with certainty.
Otherwise, if $c = 1$, then the input $y$ determines which one of three sets of single-qubit Pauli measurements he performs.
Specifically, Bob's output bits are generated as the measurement outcomes of the set of Pauli observables,
\begin{subequations}
\label{eq:bob_honest_observables}
\begin{alignat}{2}
    & \big\{ \knownoperator{X}_{B}^{1}, \knownoperator{Z}_{B}^{2}, \knownoperator{Z}_{B}^{3} \big\} && \quad \text{if $y = 1$,} \\
    & \big\{ \knownoperator{Z}_{B}^{1}, \knownoperator{X}_{B}^{2}, \knownoperator{Z}_{B}^{3} \big\} && \quad \text{if $y = 2$,} \\
    & \big\{ \knownoperator{Z}_{B}^{1}, \knownoperator{Z}_{B}^{2}, \knownoperator{X}_{B}^{3} \big\} && \quad \text{if $y = 3$.}
\end{alignat}
\end{subequations}

It is convenient at this point to call attention to the perfect correlations of output bits expected in honest check rounds.
These are all the single-qubit quantum correlations $\braket{\Psi | \knownoperator{X}_{A}^{j} \knownoperator{X}_{B}^{j} | \Psi} = 1$ and $\braket{\Psi | \knownoperator{Z}_{A}^{j} \knownoperator{Z}_{B}^{j} | \Psi} = 1$.
Observation of a version of these correlations using \emph{untrusted} observables (which will not be assumed to be identical for Bob upon his different inputs) will become a requirement for our protocol to certify the desired reference state.

\subsection{Unknown observables and correlations}
\label{sec:unknown_observables}

We will denote the unknown state shared by the players by $\ket{\Psi}$, and the expectation value of an unknown observable $M$ with respect to this state by $\langle M \rangle = \braket{\Psi | M | \Psi}$.
We now describe the unknown observables which will be used by Alice and Bob in our self-testing proof.
Recall that, in contrast to the honest Pauli observables used in the previous \cref{sec:structure_honest}, such unknown observables are denoted without a hat symbol (using $X$ for the corresponding unknown version of the Pauli $\knownoperator{X}$ observable).
We may not assume a priori, in the potentially dishonest case of the self-testing protocol, that the players measure any of the same observables upon being given different inputs.
For this reason, we introduce notation in such a way that the observer and their input can always be deduced from the label of an unknown observable.
This choice of notation will be seen in \cref{eq:alice_observables,eq:bob_game_observables,eq:bob_check_observables}.

It is important to note that all unknown observables that are to be measured as part of the same local input commute by definition.
For example, from the observables defined immediately below, it can always be assumed that $\big[ X_{A}^{1}, X_{A}^{2} \big] = 0$, since both observables correspond to the input $x = 1$ for Alice.
Furthermore, it can always be assumed that any two observables defined for different players commute.
These two properties will be exploited frequently in proofs throughout the rest of the work.

\paragraph{Alice's observables.}
We define sets of mutually commuting unknown observables on Alice's side to be measured depending on her input $x$ as
\begin{subequations}
\label{eq:alice_observables}
\begin{alignat}{2}
    & \big\{ X_{A}^{1}, X_{A}^{2}, X_{A}^{3} \big\} && \quad \text{if $x=1$,} \\
    & \big\{ Y_{A}^{1}, Y_{A}^{2}, Y_{A}^{3} \big\} && \quad \text{if $x=2$,} \\
    & \big\{ Z_{A}^{1}, Z_{A}^{2}, Z_{A}^{3} \big\} && \quad \text{if $x=3$.}
\end{alignat}
\end{subequations}
Each of these unknown observables corresponds to a single-qubit Pauli observable, which acts on the qubit of Alice indicated by its superscript.

\paragraph{Bob's observables (game rounds).}
For game rounds ($c = 0$), we will denote the sets of unknown observables to be measured by Bob, depending on his input $y$, by
\begin{subequations}
\label{eq:bob_game_observables}
\begin{alignat}{2}
    & \big\{ X_{B}^{\notqubit{1}}, Y_{B}^{\notqubit{1}}, Z_{B}^{\notqubit{1}} \big\} && \quad \text{if $y=1$,} \\
    & \big\{ X_{B}^{\notqubit{2}}, Y_{B}^{\notqubit{2}}, Z_{B}^{\notqubit{2}} \big\} && \quad \text{if $y=2$,} \\
    & \big\{ X_{B}^{\notqubit{3}}, Y_{B}^{\notqubit{3}}, Z_{B}^{\notqubit{3}} \big\} && \quad \text{if $y=3$.}
\end{alignat}
\end{subequations}
The overline notation used in each superscript reflects that these observables correspond to the product of single-qubit Pauli observables acting on all qubits of Bob other than that indicated.
For example, here the unknown observable $X_{B}^{\notqubit{1}}$ corresponds to $\knownoperator{X}_{B}^{2} \knownoperator{X}_{B}^{3}$ in the honest case.
Note also that \cref{rule:bob_prod} of the magic square game requires columns to have negative products.
In terms of unknown observables, that is $\langle X_{B}^{\notqubit{y}} Y_{B}^{\notqubit{y}} Z_{B}^{\notqubit{y}} \rangle = -1$ for all $y$.
Thus we need not have defined one observable in each set, say $Y_{B}^{\notqubit{y}}$, since this implies
\begin{equation}
\label{eq:game_y_replacement}
    Y_{B}^{\notqubit{y}} \ket{\Psi} = - X_{B}^{\notqubit{y}} Z_{B}^{\notqubit{y}} \ket{\Psi} .
\end{equation}
We will, however, choose to keep all of these observables for notational convenience, referring to \cref{eq:game_y_replacement} when necessary.

\paragraph{Bob's observables (check rounds).}
For check rounds ($c = 1$), Bob's unknown observables correspond to single-qubit Pauli $X$ and $Z$ observables acting on his system.
These will be denoted as follows, with an additional subscript to distinguish unknown observables of different inputs:
\begin{subequations}
\label{eq:bob_check_observables}
\begin{alignat}{2}
    & \big\{ X_{B,1}^{1}, Z_{B,1}^{2}, Z_{B,1}^{3} \big\} && \quad \text{if $y=1$,} \\
    & \big\{ Z_{B,2}^{1}, X_{B,2}^{2}, Z_{B,2}^{3} \big\} && \quad \text{if $y=2$,} \\
    & \big\{ Z_{B,3}^{1}, Z_{B,3}^{2}, X_{B,3}^{3} \big\} && \quad \text{if $y=3$.}
\end{alignat}
\end{subequations}

\paragraph{Correlations.}
The correlations of unknown observables amounting to a uniformly $\varepsilon_0$-close to perfect strategy for the magic square game (i.e. correlations obtained in game rounds) are, for all distinct $i,j,k \in \{1,2,3\}$,
\begin{subequations}
\label{eq:game_correlations}
\begin{align}
    \label{eq:game_correlations_x}
    \big\langle X_{A}^{i} X_{A}^{j} X_{B}^{\notqubit{k}} \big\rangle &\geq 1 - \varepsilon_0 , \\
    \label{eq:game_correlations_y}
    -\big\langle Y_{A}^{i} Y_{A}^{j} X_{B}^{\notqubit{k}} Z_{B}^{\notqubit{k}} \big\rangle &\geq 1 - \varepsilon_0 , \\
    \label{eq:game_correlations_z}
    \big\langle Z_{A}^{i} Z_{A}^{j} Z_{B}^{\notqubit{k}} \big\rangle &\geq 1 - \varepsilon_0 .
\end{align}
\end{subequations}
The correlations constituting uniformly $\varepsilon_1$-close to perfect check rounds are, again for all distinct $i,j \in \{1,2,3\}$,
\begin{subequations}
\label{eq:check_correlations}
\begin{align}
    \label{eq:check_correlations_x}
    \big\langle X_{A}^{i} X_{B,i}^{i} \big\rangle &\geq 1 - \varepsilon_1 , \\
    \label{eq:check_correlations_z}
    \big\langle Z_{A}^{i} Z_{B,j}^{i} \big\rangle &\geq 1 - \varepsilon_1 .
\end{align}
\end{subequations}
\Cref{fig:unknown_magic_square} clarifies the meaning of our unknown observables for game rounds.

\begin{figure}
    \centering
    \begin{subfigure}{0.5\linewidth}
        \centering
        \includegraphics{figures/figure3a.tikz}
        \caption{Alice's strategy.}
    \end{subfigure}\hfill
    \begin{subfigure}{0.5\linewidth}
        \centering
        \includegraphics{figures/figure3b.tikz}
        \caption{Bob's strategy.}
    \end{subfigure}
    \caption{The layout of unknown observables in a magic square strategy for (a) Alice and (b) Bob.}
    \label{fig:unknown_magic_square}
\end{figure}

\subsection{Commutation and anticommutation relations}
\label{sec:commutations}

In this section, we prove commutation and anticommutation relations (acting on our unknown state) for those unknown observables of Alice and Bob corresponding to single-qubit Pauli measurements.
To do this, we use the correlations of \cref{sec:unknown_observables}.
The results of this section are summarized in the following theorem:
\begin{theorem}
\label{thm:commutation_summary}
    Let $i,j,k,l \in \{1,2,3\}$ be such that $i \neq k$ and $j \neq l$.
    We have correlations between each unknown observable of Alice with each of the corresponding observables on Bob's side
    \begin{align}
        \big\lVert \big( X_{A}^{i} - X_{B,i}^{i} \big) \ket{\Psi} \big\rVert &\leq \sqrt{2\varepsilon_1} , \\
        \big\lVert \big( Z_{A}^{i} - Z_{B,k}^{i} \big) \ket{\Psi} \big\rVert &\leq \sqrt{2\varepsilon_1} .
    \end{align}
    We have the state-dependent anticommutativity of all unknown $X$ observables with all unknown $Z$ observables corresponding to the same qubit
    \begin{align}
        \big\lVert \big\{ X_{A}^{i} , Z_{A}^{i} \big\} \ket{\Psi} \big\rVert &\leq 9\sqrt{2\varepsilon_0} + 16\sqrt{2\varepsilon_1} , \\
        \big\lVert \big\{ X_{B,i}^{i}, Z_{B,k}^{i} \big\} \ket{\Psi} \big\rVert &\leq 9\sqrt{2\varepsilon_0} + 20\sqrt{2\varepsilon_1} .
    \end{align}
    Finally, we have the state-dependent commutativity of unknown $X$ and $Z$ observables.
    On Bob's side we have
    \begin{align}
        \big\lVert \big[ X_{B,i}^{i}, X_{B,j}^{j} \big] \ket{\Psi} \big\rVert &\leq 4\sqrt{2\varepsilon_1} , \\
        \big\lVert \big[ Z_{B,k}^{i}, Z_{B,l}^{j} \big] \ket{\Psi} \big\rVert &\leq 4\sqrt{2\varepsilon_1} ;
    \end{align}
    and moreover restricting to observables corresponding to different qubits $i \neq j$
    \begin{equation}
        \big\lVert \big[ X_{B,i}^{i}, Z_{B,l}^{j} \big] \ket{\Psi} \big\rVert \leq 8\sqrt{2\varepsilon_1} .
    \end{equation}
    On Alice's side, for different qubits $i \neq j$, we have
    \begin{equation}
        \big\lVert \big[ M_{A}^{i}, N_{A}^{j} \big] \ket{\Psi} \big\rVert \leq 4\sqrt{2\varepsilon_1} ,
    \end{equation}
    where $M$ and $N$ can be either of $X$ and $Z$.
\end{theorem}
\begin{proof}
    Combine \cref{prop:correlation_estimation,prop:commutations,prop:anticommutations}.
\end{proof}

\begin{proposition}[Correlation]
\label{prop:correlation_estimation}
    For all distinct $i,j \in \{1,2,3\}$ we have the correlation estimates
    \begin{subequations}
    \begin{align}
        \big\lVert \big( X_{A}^{i} - X_{B,i}^{i} \big) \ket{\Psi} \big\rVert &\leq \sqrt{2\varepsilon_1} , \\
        \big\lVert \big( Z_{A}^{i} - Z_{B,j}^{i} \big) \ket{\Psi} \big\rVert &\leq \sqrt{2\varepsilon_1} .
    \end{align}
    \end{subequations}
\end{proposition}
\begin{proof}
    Apply \cref{lem:state_estimate} to the correlations given in \cref{eq:check_correlations}.
\end{proof}

The following proposition shows the commutation of unknown observables which we expect to correspond to local measurements on different qubits.
Since observables defined for different players are assumed to commute, we show commutation for the observables of each player separately.
\begin{proposition}[Commutation]
\label{prop:commutations}
    For all $i,j,k,l \in \{1,2,3\}$ such that $i \neq k$ and $j \neq l$ we have
    \begin{subequations}
    \begin{align}
        \label{eq:bob_comm_xx}
        \big\lVert \big[ X_{B,i}^{i}, X_{B,j}^{j} \big] \ket{\Psi} \big\rVert &\leq 4\sqrt{2\varepsilon_1} , \\
        \label{eq:bob_comm_zz}
        \big\lVert \big[ Z_{B,k}^{i}, Z_{B,l}^{j} \big] \ket{\Psi} \big\rVert &\leq 4\sqrt{2\varepsilon_1} .
    \end{align}
    \end{subequations}
    Moreover if $i \neq j$ we have commutation relations for Bob
    \begin{equation}
        \label{eq:bob_comm_xz}
        \big\lVert \big[ X_{B,i}^{i}, Z_{B,l}^{j} \big] \ket{\Psi} \big\rVert \leq 8\sqrt{2\varepsilon_1}
    \end{equation}
    and for Alice
    \begin{equation}
    \label{eq:alice_comm}
        \big\lVert \big[ M_{A}^{i}, N_{A}^{j} \big] \ket{\Psi} \big\rVert \leq 4\sqrt{2\varepsilon_1} ,
    \end{equation}
    where $M$ and $N$ can be either of $X$ and $Z$.
\end{proposition}
\begin{proof}
    Using the triangle inequality with the estimates of \cref{prop:correlation_estimation}, and the commutation of Alice's observables corresponding to the same input, we can write
    \begin{equation}
    \begin{split}
        \big\lVert \big[ X_{B,i}^{i}, X_{B,j}^{j} \big] \ket{\Psi} \big\rVert
        &\leq 4\sqrt{2\varepsilon_1} + \big\lVert \big[ X_{A}^{j}, X_{A}^{i} \big] \ket{\Psi} \big\rVert \\
        &= 4\sqrt{2\varepsilon_1} ,
    \end{split}
    \end{equation}
    showing \cref{eq:bob_comm_xx}.
    Similarly, to obtain \cref{eq:bob_comm_zz},
    \begin{equation}
    \begin{split}
        \big\lVert \big[ Z_{B,k}^{i}, Z_{B,l}^{j} \big] \ket{\Psi} \big\rVert
        &\leq 4\sqrt{2\varepsilon_1} + \big\lVert \big[ Z_{A}^{j}, Z_{A}^{i} \big] \ket{\Psi} \big\rVert \\
        &= 4\sqrt{2\varepsilon_1} .
    \end{split}
    \end{equation}

    We now assume $i \neq j$.
    From the definition of Bob's check-round observables [\cref{eq:bob_check_observables}] we have $\big[ X_{B,i}^{i}, Z_{B,i}^{j} \big] = 0$.
    We use this and \cref{prop:correlation_estimation} to get
    \begin{equation}
    \begin{split}
        \big\lVert \big[X_{A}^{i}, Z_{A}^{j} \big] \ket{\Psi} \big\rVert
        &= \big\lVert X_{A}^{i} Z_{A}^{j} \ket{\Psi} - Z_{A}^{j} X_{A}^{i} \ket{\Psi} \big\rVert \\
        &\leq 2\sqrt{2\varepsilon_1} + \big\lVert X_{A}^{i} Z_{A}^{j} \ket{\Psi} - X_{B,i}^{i} Z_{B,i}^{j} \ket{\Psi} \big\rVert \\
        &= 2\sqrt{2\varepsilon_1} + \big\lVert X_{A}^{i} Z_{A}^{j} \ket{\Psi} - Z_{B,i}^{j} X_{B,i}^{i} \ket{\Psi} \big\rVert \\
        &\leq 4\sqrt{2\varepsilon_1} + \big\lVert X_{A}^{i} Z_{A}^{j} \ket{\Psi} - X_{A}^{i} Z_{A}^{j} \ket{\Psi} \big\rVert \\
        &= 4\sqrt{2\varepsilon_1} .
    \end{split}
    \end{equation}
    Combining this with the definition of Alice's observables [\cref{eq:alice_observables}], from which we have $\big[ X_{A}^{i}, X_{A}^{j} \big] = 0$ and $\big[ Z_{A}^{i}, Z_{A}^{j} \big] = 0$, yields \cref{eq:alice_comm}.
    To obtain \cref{eq:bob_comm_xz}, we again use \cref{prop:correlation_estimation} to write
    \begin{equation}
        \big\lVert \big[ X_{B,i}^{i}, Z_{B,l}^{j} \big] \ket{\Psi} \big\rVert \leq 4\sqrt{2\varepsilon_1} + \big\lVert \big[ Z_{A}^{j}, X_{A}^{i} \big] \ket{\Psi} \big\rVert \leq 8\sqrt{2\varepsilon_1} ,
    \end{equation}
    where the final inequality uses \cref{eq:alice_comm} just proved.
\end{proof}

We now show an intermediate result that will allow us to prove the anticommutativity of unknown local $X$ and $Z$ observables.
The lemma shows that Alice's unknown observables for pairs of $X$ and $Z$ operators not acting on the same qubits anticommute (cf. the observables used in the magic square strategy of \cref{sec:one_side_local_strategy}).
The proof follows a similar line to \cite{wu2016device}.
\begin{lemma}
\label{lem:alice_pair_anticomm}
    For all distinct $i,j,k \in \{1,2,3\}$ we have anticommutation relations for Bob's game round observables
    \begin{equation}
    \label{eq:alice_pair_anticomm}
        \big\lVert \big\{ X_{A}^{i} X_{A}^{j} , Z_{A}^{i} Z_{A}^{k} \big\} \ket{\Psi} \big\rVert \leq 9\sqrt{2\varepsilon_0} .
    \end{equation}
\end{lemma}
\begin{proof}
    By estimating the game-round correlations of \cref{eq:game_correlations} using \cref{lem:state_estimate}, and repeatedly applying the triangle inequality,
    \begin{equation}
    \begin{split}
        \big\lVert \big\{ X_{A}^{i} X_{A}^{j} , Z_{A}^{i} Z_{A}^{k} \big\} \ket{\Psi} \big\rVert
        &\leq 4\sqrt{2\varepsilon_0} + \Big\lVert Z_{B}^{\notqubit{j}} X_{B}^{\notqubit{k}} \ket{\Psi} + X_{B}^{\notqubit{j}} X_{B}^{\notqubit{i}} Z_{A}^{i} Z_{A}^{k} \ket{\Psi} \Big\rVert \\
        &= 4\sqrt{2\varepsilon_0} + \Big\lVert X_{B}^{\notqubit{j}} Z_{B}^{\notqubit{j}} X_{B}^{\notqubit{k}} Z_{A}^{i} Z_{A}^{j} \ket{\Psi} + X_{B}^{\notqubit{i}} Z_{A}^{j} Z_{A}^{k} \ket{\Psi} \Big\rVert \\
        &\leq 6\sqrt{2\varepsilon_0} + \Big\lVert \Big( X_{B}^{\notqubit{j}} Z_{B}^{\notqubit{j}} \Big) \Big( X_{B}^{\notqubit{k}} Z_{B}^{\notqubit{k}} \Big) \ket{\Psi} + X_{B}^{\notqubit{i}} Z_{B}^{\notqubit{i}} \ket{\Psi} \Big\rVert \\
        &\leq 8\sqrt{2\varepsilon_0} + \Big\lVert \big( Y_{A}^{i} Y_{A}^{j} \big) \big( Y_{A}^{i} Y_{A}^{k} \big) \ket{\Psi} + X_{B}^{\notqubit{i}} Z_{B}^{\notqubit{i}} \ket{\Psi} \Big\rVert \\
        &= 8\sqrt{2\varepsilon_0} + \Big\lVert Y_{A}^{j} Y_{A}^{k} \ket{\Psi} + X_{B}^{\notqubit{i}} Z_{B}^{\notqubit{i}} \ket{\Psi} \Big\rVert \\
        &\leq 9\sqrt{2\varepsilon_0} ,
    \end{split}
    \end{equation}
    where the first equality results from applying unitary operators $Z_{A}^{i} Z_{A}^{j}$ and $X_{B}^{\notqubit{j}}$ inside the norm.
\end{proof}

We are now in a position to prove the required anticommutativity of unknown $X$ observables with $Z$ observables which act on the same qubits of the unknown state.
\begin{proposition}[Anticommutation]
\label{prop:anticommutations}
    For all $i \in \{1,2,3\}$ we have anticommutation relations for Alice's unknown observables
    \begin{equation}
    \label{eq:alice_anticomm}
        \big\lVert \big\{ X_{A}^{i} , Z_{A}^{i} \big\} \ket{\Psi} \big\rVert
        \leq 9\sqrt{2\varepsilon_0} + 16\sqrt{2\varepsilon_1} .
    \end{equation}
    Furthermore, for all $j \in \{1,2,3\}$ distinct from $i$ we have anticommutation relations for Bob's check-round observables
    \begin{equation}
    \label{eq:bob_check_anticomm}
        \big\lVert \big\{ X_{B,i}^{i}, Z_{B,j}^{i} \big\} \ket{\Psi} \big\rVert
        \leq 9\sqrt{2\varepsilon_0} + 20\sqrt{2\varepsilon_1} .
    \end{equation}
\end{proposition}
\begin{proof}
    Let $k \in \{1,2,3\}$ be distinct from $i$ and $j$, then
    \begin{equation}
    \label{eq:alice_anticomm_transformed}
    \begin{split}
        \big\lVert \big\{ X_{A}^{i} , Z_{A}^{i} \big\} \ket{\Psi} \big\rVert
        &= \big\lVert X_{B,j}^{j} Z_{B,i}^{k} \big\{ X_{A}^{i} , Z_{A}^{i} \big\} \ket{\Psi} \big\rVert \\
        &= \big\lVert X_{A}^{i} Z_{A}^{i} X_{B,j}^{j} Z_{B,i}^{k} \ket{\Psi} + Z_{A}^{i} X_{A}^{i} X_{B,j}^{j} Z_{B,i}^{k} \ket{\Psi} \big\rVert \\
        &\leq \big\lVert X_{A}^{i} Z_{A}^{i} Z_{B,i}^{k} X_{B,j}^{j} \ket{\Psi} + Z_{A}^{i} X_{A}^{i} X_{B,j}^{j} Z_{B,i}^{k} \ket{\Psi} \big\rVert
        + 8\sqrt{2\varepsilon_1} \\
        &\leq \big\lVert X_{A}^{i} Z_{B,i}^{k} X_{B,j}^{j} Z_{B,j}^{i} \ket{\Psi} + Z_{A}^{i} X_{B,j}^{j} Z_{B,i}^{k} X_{B,i}^{i} \ket{\Psi} \big\rVert
        + 10\sqrt{2\varepsilon_1} \\
        &= \big\lVert X_{A}^{i} Z_{B,i}^{k} Z_{B,j}^{i} X_{B,j}^{j} \ket{\Psi} + Z_{A}^{i} X_{B,j}^{j} X_{B,i}^{i} Z_{B,i}^{k} \ket{\Psi} \big\rVert
        + 10\sqrt{2\varepsilon_1} \\
        &\leq \big\lVert \big\{ X_{A}^{i} X_{A}^{j} , Z_{A}^{i} Z_{A}^{k} \big\} \ket{\Psi} \big\rVert + 16\sqrt{2\varepsilon_1} \\
        &\leq 9\sqrt{2\varepsilon_0} + 16\sqrt{2\varepsilon_1} .
    \end{split}
    \end{equation}
    For the first inequality, we commuted Bob's check-round observables using \cref{eq:bob_comm_xz} of \cref{prop:commutations}.
    For the final inequality, we applied \cref{lem:alice_pair_anticomm} to bound the anticommutator norm.
    All other inequalities were found from the correlation estimates of \cref{prop:correlation_estimation}.

    To obtain \cref{eq:bob_check_anticomm} we use \cref{prop:correlation_estimation} to write
    \begin{equation}
    \begin{split}
        \big\lVert \big\{ X_{B,i}^{i}, Z_{B,j}^{i} \big\} \ket{\Psi} \big\rVert
        &\leq 4\sqrt{2\varepsilon_1} + \big\lVert \big\{ X_{A}^{i}, Z_{A}^{i} \big\} \ket{\Psi} \big\rVert \\
        &\leq 9\sqrt{2\varepsilon_0} + 20\sqrt{2\varepsilon_1} ,
    \end{split}
    \end{equation}
    where the final inequality follows from \cref{eq:alice_anticomm} just proved.
\end{proof}

\section{Self-test of many Bell states}
\label{sec:selftest_n}

We can use similar techniques to \cref{sec:selftest_3} to self-test $n > 3$ Bell states, provided $n \equiv 3 \pmod{4}$ (which we will assume throughout this section).
In this case, the honest strategy is played using a \rectdim{3}{n} magic game, as described by \cref{def:3xn_magic_rectangle}.
The strategy for this game upon which we base our self-test will be explained in \cref{sec:3xn_magic_game_strategy}.
The structure and honest behavior of the self-test will simultaneously be described in \cref{sec:3xn_honest}, with all general unknown observables for Alice and Bob and their required correlations then defined in \cref{sec:3xn_correlations}.
All commutation and anticommutation relations required to construct a local self-testing isometry will finally be shown in \cref{sec:3xn_commutations}.
From this, we have the final self-testing statement for many Bell states.
\begin{theorem}
\label{thm:3xn_isometry}
    Let $\ket{\Psi} \in \mathcal{H}_{A} \otimes \mathcal{H}_{B}$ be an unknown state shared by Alice and Bob and let $n \equiv 3 \pmod{4}$ with $n > 3$ be the number of Bell states to be self-tested.
    Let sets of pairwise commutative, $\pm 1$-valued, unknown observables in $\mathcal{L}(\mathcal{H}_{A})$ for Alice be given as in \cref{eq:3xn_alice_observables}, and in $\mathcal{L}(\mathcal{H}_{B})$ for Bob as in \cref{eq:3xn_bob_game_observables,eq:local_check_observables,eq:pair_check_observables}.
    Suppose that these observables satisfy all correlations given in \cref{eq:3xn_game_correlations,eq:local_check_correlations,eq:pair_check_correlations} and let $\varepsilon = \max \{ \varepsilon_0, \varepsilon_1, \varepsilon_2 \}$.
    Then, for any choice $(k_i)_{i=1}^{n}$ of elements in $\{1,\dots,n\}$ where each $k_i \neq i$, there exists a junk state $\ket{\xi}$ and a local isometry $\Phi$ such that, for all $i \in \{1,\dots,n\}$,
    \begin{subequations}
    \begin{align}
        \big\lVert \Phi \ket{\Psi} - \ket{\Phi^{+}}_{AB}^{\otimes n} \otimes \ket{\xi} \big\rVert
        &\in \bigO \big( n^{\frac{5}{2}} \sqrt{\varepsilon} \big) , \\
        \big\lVert \Phi X_{A}^{i} \ket{\Psi} - \knownoperator{X}_{A}^{i} \ket{\Phi^{+}}_{AB}^{\otimes n} \otimes \ket{\xi} \big\rVert
        &\in \bigO \big( n^{\frac{5}{2}} \sqrt{\varepsilon} \big) , \\
        \big\lVert \Phi Z_{A}^{i} \ket{\Psi} - \knownoperator{Z}_{A}^{i} \ket{\Phi^{+}}_{AB}^{\otimes n} \otimes \ket{\xi} \big\rVert
        &\in \bigO \big( n^{\frac{5}{2}} \sqrt{\varepsilon} \big) , \\
        \big\lVert \Phi X_{B,i}^{i} \ket{\Psi} - \knownoperator{X}_{B}^{i} \ket{\Phi^{+}}_{AB}^{\otimes n} \otimes \ket{\xi} \big\rVert
        &\in \bigO \big( n^{\frac{5}{2}} \sqrt{\varepsilon} \big) , \\
        \big\lVert \Phi Z_{B,k_i}^{i} \ket{\Psi} - \knownoperator{Z}_{B}^{i} \ket{\Phi^{+}}_{AB}^{\otimes n} \otimes \ket{\xi} \big\rVert
        &\in \bigO \big( n^{\frac{5}{2}} \sqrt{\varepsilon} \big) .
    \end{align}
    \end{subequations}
\end{theorem}
\begin{proof}
    Take the observables $\{ X_{A}^{i}, Z_{A}^{i} \}_{i=1}^{n}$ of \cref{eq:3xn_alice_observables} and $\{ X_{B,i}^{i}, Z_{B,k_i}^{i} \}_{i=1}^{n}$ of \cref{eq:local_check_observables} to be the (extended if necessary) reflections assumed by \cref{thm:isometry}, with $\delta$ given by the largest upper bound appearing in \cref{thm:3xn_commutation_summary}.
\end{proof}
Relatively few of the unknown observables defined as part of the self-test are actually used to construct the isometry, with most only serving in the proofs of necessary commutation and anticommutation relations.
The total number of observables defined in \cref{eq:3xn_alice_observables,eq:3xn_bob_game_observables,eq:local_check_observables,eq:pair_check_observables} is $2n^2 + 4n$, while only $4n$ of these are required for the isometry of \cref{thm:isometry}.
In particular, we are free to use any $n$ of the $Z_{B,y}^{i}$ of \cref{eq:local_check_observables} provided that we cover all qubits (denoted by the superscript index).
This freedom is expressed in \cref{thm:3xn_isometry} above by choice of the $k_i$.
In the honest case, many of the unknown observables are in fact identical to one another.

For this self-test, Bob must make Pauli measurements on pairs of qubits to ensure their commutation.
This was not explicitly required in the self-test of three Bell states, since Bob's game-round observables (corresponding to products of Pauli observables on all but one of his qubits) automatically served this purpose.
We would thus like a way to subdivide all possible pairs of (an odd number of) qubits into as few disjoint sets of disjoint pairs as possible.
This is equivalent to finding an optimal edge coloring for the complete graph $K_n$ where $n$ is odd.
The following lemma constructs such a coloring.
\begin{lemma}
\label{lem:complete_graph_coloring}
    Consider the complete graph $K_{n}$ for $n$ odd, whose vertices are labeled by $V = \{1,\dots,n\}$.
    For each $v \in V$, color the edges $\{v-i, v+i\}$ by color $v$ for all $i \in \big\{ 1, \dots, \frac{n-1}{2} \big\}$, where addition is performed modulo $n$.
    This is a proper $n$-edge-coloring for $K_n$ and is optimal in the sense that it uses as few colors as possible.
\end{lemma}
\begin{proof}
    Define the color of each edge $\{a, b\}$ to be $\frac{a+b}{2} \pmod{n}$, where the multiplicative inverse of $2$ modulo $n$ always exists since $2$ is coprime to any odd $n$.
    Suppose that two edges $\{x, i\}$ and $\{x, j\}$ have the same color under this definition.
    Then $\frac{x + i}{2} \equiv \frac{x + j}{2} \pmod{n}$, and thus $i = j$.
    Therefore no two distinct adjacent edges can have the same color.
    That is, we defined a proper edge coloring.
    Notice that all edges of the same color $v$ here take the form $\{v-i, v+i\}$ for $i \in \big\{ 1, \dots, \frac{n-1}{2} \big\}$.
    Hence our coloring is identical to that given in the statement.
    Optimality results from the fact that the chromatic index of $K_n$ is $n$ when $n$ is odd.
\end{proof}
\begin{remark}
    If the graph is depicted by straight lines drawn between the vertices of a regular $n$-gon, the given construction assigns a different color to each of $n$ sets of parallel edges.
\end{remark}

Since we will be dealing with many noncommutative objects, we unambiguously define the finite product notation to be formed with indices in ascending order as
\begin{equation}
    \prod_{i=1}^{n} M_i \equiv M_1 M_2 \dots M_n .
\end{equation}
We will use this notation to denote the composition of (not necessarily commutative) operators.

\subsection{Magic game strategy}
\label{sec:3xn_magic_game_strategy}

A simple winning strategy for \rectdim{3}{n} magic games, in which players share three Bell states and Alice need only make single-qubit measurements, can be constructed by appending deterministic columns to the \rectdim{3}{3} strategy of \cref{sec:one_side_local_strategy}.
However, we will base our self-test on an alternative strategy, which will be described here.

Let Alice and Bob share the $n$ Bell states
\begin{equation}
    \ket{\Psi} = \bigotimes_{j=1}^{n} \ket{\Phi^{+}}_{AB}^{(j)} .
\end{equation}
\Cref{fig:3xn_strategy} depicts the \rectdim{3}{n} measurement strategy which our self-test will be based on.

\begin{figure}
    \centering
    \includegraphics{figures/figure4.tikz}
    \caption{
        The \rectdim{3}{n} magic game strategy that our self-test is based upon.
        Pauli observables which act on qubit $j$ of a player's register are denoted by $\knownoperator{X}^j$, $\knownoperator{Y}^j$, and $\knownoperator{Z}^j$.
        }
    \label{fig:3xn_strategy}
\end{figure}

Since $n \equiv 3 \pmod{4}$, the observable for each square of the strategy is composed of $2 \pmod{4}$ single-qubit Pauli observables.
Hence the three observables in each column mutually commute and satisfy Bob's negative product rule
\begin{equation}
\label{eq:3xn_strategy_bob_prod}
    \Bigg( \prod_{j \neq y} \knownoperator{X}^j \Bigg) \Bigg( \prod_{j \neq y} \knownoperator{Y}^j \Bigg) \Bigg( \prod_{j \neq y} \knownoperator{Z}^j \Bigg)
    = \prod_{j \neq y} \knownoperator{X}^j \knownoperator{Y}^j \knownoperator{Z}^j = i^{n-1} I = i^2 I = -I .
\end{equation}
Since the Pauli observables appearing in each row are all of the same type, the squares in each row mutually commute.
Moreover, since Pauli observables are involutory and there are an even number of such observables corresponding to each qubit in each row, every row has product $+I$.
There is also perfect correlation between Alice's and Bob's observables for each square of the strategy.
That is, letting $\knownoperator{S}$ stand for $\knownoperator{X}$, $\knownoperator{Y}$, or $\knownoperator{Z}$, and for all $y$,
\begin{equation}
\label{eq:3xn_strategy_correlation}
    \braket{\Psi | \prod_{j \neq y} \knownoperator{S}_{A}^{j} \prod_{j \neq y} \knownoperator{S}_{B}^{j} | \Psi}
    = \prod_{j \neq y} \braket{\Phi^{+} | \knownoperator{S}_{A}^{j} \knownoperator{S}_{B}^{j} | \Phi^{+}}_{AB}^{(j)}
    = (\pm 1)^{n-1} = 1 .
\end{equation}

This strategy can again be naturally phrased as a winning strategy for magic games in the sense of \cref{def:3xn_magic_rectangle}.
Alice generates her outputs $a_j$ as the outcomes of measurements of $\knownoperator{X}_{A}^{j}$, $\knownoperator{Y}_{A}^{j}$, or $\knownoperator{Z}_{A}^{j}$ depending on whether the first, second, or third row was assigned respectively.
Bob generates his outputs $(b_1, b_2, b_3)$ according to the outcomes of observables in \cref{fig:3xn_strategy} for the column he was assigned.
By \cref{eq:3xn_strategy_bob_prod}, Bob's outputs always satisfy the rule $b_1 b_2 b_3 = -1$.
By \cref{eq:3xn_strategy_correlation}, for input row and columns $x$ and $y$, respectively, the outputs always satisfy $\prod_{j \neq y} a_j = b_x$.
Therefore, in the strategy described, the players win with certainty.

In terms of experimental implementation, note that Alice need only make single-qubit Pauli measurements for her side of the strategy.
On Bob's side, making the required compatible measurements of $\prod_{j \neq y} \knownoperator{X}_{B}^{j}$, $\prod_{j \neq y} \knownoperator{Y}_{B}^{j}$, and $\prod_{j \neq y} \knownoperator{Z}_{B}^{j}$ may seem impractical for systems with large $n$.
Note, however, that since the pairs of Pauli observables $\knownoperator{X} \otimes \knownoperator{X}$, $\knownoperator{Y} \otimes \knownoperator{Y}$, and $\knownoperator{Z} \otimes \knownoperator{Z}$ mutually commute, Bob need only measure $\frac{3}{2} (n - 1)$ such pairs to construct measurements of all three required observables.

\subsection{Structure and honest behavior}
\label{sec:3xn_honest}

As in our self-test for three Bell states, Alice receives an input $x \in \{1,2,3\}$.
However, Bob now receives an input $y \in \{1,\dots,n\}$.
Furthermore, Bob's input controlling the type of round is now a trit $c \in \{0,1,2\}$.
The additional value $c = 2$ determines that the players are requested to check correlations between certain pairs of Pauli observables.
As such, we will call such rounds where $c=2$ \emph{pair check} rounds, and rename those rounds where $c=1$ to \emph{local check} rounds to avoid ambiguity.
Alice must always output $n$ bits, whereas the number of output bits of Bob depends on the type of round $c$.
The protocol is summarized in \cref{prot:protocol_n}.

\begin{protocol}[htb]
    \caption{
        Protocol for certifying $n$ Bell states.
        Intuitively, \emph{pair check} rounds rule out those single-qubit \rectdim{3}{n} magic rectangle game strategies found by extending strategies for smaller \rectdim{3}{n^{\prime}} games using deterministic entries.
        Otherwise, the required correlations could be satisfied by provers sharing fewer Bell states.
    }
    \label{prot:protocol_n}
    \rule{\linewidth}{0.08em}
    Let $n = 3 \pmod{4}$ be the number of Bell states to be certified.
    The verifier chooses $c \in \{0,1,2\}$ and performs \cref{prot:protocol_3} with an additional subprotocol if $c=2$ is chosen:
    \begin{enumerate}
        \setcounter{enumi}{1}
        \item \emph{Pair check}:
        Send Alice $x \in \{1,3\}$.
        Alice answers with $a_{1}, \dots, a_{n}$.
        Bob answers with $n-1$ bits $b_{y-k,y+k}$ and $b^{\prime}_{y-k,y+k}$ in $\{+1,-1\}$ (with addition taken modulo $n$) for all $k \in \{1,\dots,\frac{n-1}{2}\}$.
        \begin{enumerate}
            \item If $x = 1$, accept if and only if $a_{i} a_{j} = b_{i,j}$ for all $i,j$.
            \item If $x = 3$, accept if and only if $a_{i} a_{j} = b^{\prime}_{i,j}$ for all $i,j$.
        \end{enumerate}
    \end{enumerate}\vspace*{-\baselineskip}
    \rule{\linewidth}{0.08em}
\end{protocol}

Honest rounds consist of the players sharing $n$ Bell states,
\begin{equation}
    \ket{\Psi} = \bigotimes_{j=1}^{n} \ket{\Phi^{+}}_{AB}^{(j)} .
\end{equation}
Alice always provides each of her output bits $a_j$ by measuring the $n$ observables of our \rectdim{3}{n} magic game strategy (\cref{sec:3xn_magic_game_strategy}),
\begin{subequations}
\begin{alignat}{2}
    & \knownoperator{X}_{A}^{j} && \quad \text{if $x = 1$,} \\
    & \knownoperator{Y}_{A}^{j} && \quad \text{if $x = 2$,} \\
    & \knownoperator{Z}_{A}^{j} && \quad \text{if $x = 3$.}
\end{alignat}
\end{subequations}
This is structurally identical to \cref{eq:alice_honest_observables} in the previous self-test of three Bell states, with the exception that $n$ measurements are now made upon each input.

Once again the honest behavior of Bob depends on $c$.
If it is a game round ($c=0$), then Bob must output three bits, as usual with the goal of winning the \rectdim{3}{n} magic game.
In the case of a local check round ($c=1$), Bob proceeds similarly to \cref{eq:bob_honest_observables} of the previous self-test, but now generates his $j$th of $n$ output bits depending on the input $y$ as the measurement outcomes of Pauli observables $\knownoperator{S}_{B}^{j}$, where
\begin{equation}
    \knownoperator{S}_{B}^{j} = \begin{cases}
        \knownoperator{X}_{B}^{j} & \text{if $y = j$,} \\
        \knownoperator{Z}_{B}^{j} & \text{otherwise.}
    \end{cases}
\end{equation}
Finally, if it is a pair check round ($c=2$), Bob measures $n-1$ Pauli observables of the form $\knownoperator{X} \otimes \knownoperator{X}$ and $\knownoperator{Z} \otimes \knownoperator{Z}$ on disjoint pairs of qubits.
Depending on the input $y$, the observables he measures are
\begin{equation}
\label{eq:honest_pair_check_observables}
    \big\{ \knownoperator{X}_{B}^{y-j} \knownoperator{X}_{B}^{y+j} \big\}_{j=1}^{(n-1)/2} \cup \big\{ \knownoperator{Z}_{B}^{y-j} \knownoperator{Z}_{B}^{y+j} \big\}_{j=1}^{(n-1)/2} ,
\end{equation}
where addition is taken modulo $n$.
Notice that all observables in the set of \cref{eq:honest_pair_check_observables} mutually commute, and by the construction given in \cref{lem:complete_graph_coloring} the combination of all $n$ such sets covers every possible pair of $n$ qubits.

The correlations that we expect to be satisfied from honest check rounds are the appropriate perfect correlations between Alice and Bob.
For local check rounds, these are (as before) all the single-qubit correlations $\braket{\Psi | \knownoperator{X}_{A}^{j} \knownoperator{X}_{B}^{j} | \Psi} = 1$ and $\braket{\Psi | \knownoperator{Z}_{A}^{j} \knownoperator{Z}_{B}^{j} | \Psi} = 1$.
For pair check rounds, these are the correlations between all pairs of observables $\braket{\Psi | \knownoperator{X}_{A}^{j} \knownoperator{X}_{A}^{k} \knownoperator{X}_{B}^{j} \knownoperator{X}_{B}^{k} | \Psi} = 1$ and $\braket{\Psi | \knownoperator{Z}_{A}^{j} \knownoperator{Z}_{A}^{k} \knownoperator{Z}_{B}^{j} \knownoperator{Z}_{B}^{k} | \Psi} = 1$.

\subsection{Unknown observables and correlations}
\label{sec:3xn_correlations}

Recall that, as in \cref{sec:unknown_observables} for the previous self-test of three Bell states, all unknown observables must be labeled uniquely with respect to each observer's possible input questions in order to avoid assumptions about their measurements in this potentially dishonest case.

\paragraph{Alice's observables.}
We define sets of mutually commuting unknown observables on Alice's side to be measured depending on her input $x$ as
\begin{subequations}
\label{eq:3xn_alice_observables}
\begin{alignat}{2}
    & \big\{ X_{A}^{j} \big\}_{j=1}^{n} && \quad \text{if $x=1$,} \\
    & \big\{ Y_{A}^{j} \big\}_{j=1}^{n} && \quad \text{if $x=2$,} \\
    & \big\{ Z_{A}^{j} \big\}_{j=1}^{n} && \quad \text{if $x=3$.}
\end{alignat}
\end{subequations}
Each of these unknown observables corresponds to a single-qubit Pauli observable which acts on the qubit of Alice indicated by its superscript.

\paragraph{Bob's observables (game rounds).}
For game rounds ($c = 0$), we will denote the sets of unknown observables to be measured by Bob, depending on his input $y$, by
\begin{equation}
\label{eq:3xn_bob_game_observables}
    \big\{ X_{B}^{\notqubit{y}}, Y_{B}^{\notqubit{y}}, Z_{B}^{\notqubit{y}} \big\} .
\end{equation}
It should once again be noted that one of the observables for each input is redundant, as
\begin{equation}
    Y_{B}^{\notqubit{y}} \ket{\Psi} = - X_{B}^{\notqubit{y}} Z_{B}^{\notqubit{y}} \ket{\Psi}
\end{equation}
by the rule for the product of Bob's outputs (see \cref{def:3xn_magic_rectangle}).
We will, however, keep all for notational convenience.

\paragraph{Bob's observables (local check rounds).}
For local check rounds ($c = 1$), Bob's unknown observables correspond to single-qubit Pauli $\knownoperator{X}$ and $\knownoperator{Z}$ observables acting on his system.
The set of observables for input $y$ is defined by
\begin{equation}
\label{eq:local_check_observables}
    \big\{ X_{B,y}^{y} \big\} \cup \big\{ Z_{B,y}^{j} : 1 \leq j \leq n, j \neq y \big\} .
\end{equation}

\paragraph{Bob's observables (pair check rounds).}
For pair check rounds ($c = 2$), we define sets of $n-1$ observables for each input $y$ as
\begin{equation}
\label{eq:pair_check_observables}
    \big\{ X_{B}^{y-j, y+j} \big\}_{j=1}^{(n-1)/2} \cup \big\{ Z_{B}^{y-j, y+j} \big\}_{j=1}^{(n-1)/2}
\end{equation}
where addition is taken modulo $n$.
In contrast to the honest case of \cref{eq:honest_pair_check_observables}, we have not assumed that Bob's outputs arise as the product of multiple other observables.
The two superscript indices denote that these observables correspond to the product of Pauli observables on pairs of qubits.
For example, the unknown observable $X_{B}^{1,2}$ corresponds to $\knownoperator{X}_{B}^{1} \knownoperator{X}_{B}^{2}$ in the honest case.
In the notation we have introduced, the order of superscript indices for an unknown observable is unimportant.
Thus, for convenience, we also introduce labels with reversed ordering of superscripts and identify these with observables appearing in \cref{eq:pair_check_observables}.
Specifically, let the labels $X_{B}^{i,j} \equiv X_{B}^{j,i}$ and $Z_{B}^{i,j} \equiv Z_{B}^{j,i}$.
This is consistent with the honest case, in which the corresponding pairs of observables commute.
By \cref{lem:complete_graph_coloring}, the pairs of indices $(y-j, y+j)$ appearing in \cref{eq:pair_check_observables} for a given input $y$ are pairwise disjoint and the combination of these pairs over every input gives every possible index pair (up to ordering of the indices).
Thus the $n$ sets of $n-1$ pair check observables defined account for measurements of $\knownoperator{X} \otimes \knownoperator{X}$ and $\knownoperator{Z} \otimes \knownoperator{Z}$ on every possible pair of $n$ qubits and, moreover, the observables for a given input mutually commute in the honest case as expected.

\paragraph{Correlations.}
The correlations of unknown observables amounting to a uniformly $\varepsilon_0$-close to perfect strategy for the \rectdim{3}{n} magic game (i.e. correlations obtained in game rounds) are, with reference to the winning strategy described in \cref{sec:3xn_magic_game_strategy},
\begin{subequations}
\label{eq:3xn_game_correlations}
\begin{align}
    \label{eq:3xn_game_correlations_x}
    \bigg\langle \bigg( \prod_{j \neq k} X_{A}^{j} \bigg) X_{B}^{\notqubit{k}} \bigg\rangle &\geq 1 - \varepsilon_0 , \\
    \label{eq:3xn_game_correlations_y}
    -\bigg\langle \bigg( \prod_{j \neq k} Y_{A}^{j} \bigg) X_{B}^{\notqubit{k}} Z_{B}^{\notqubit{k}} \bigg\rangle &\geq 1 - \varepsilon_0 , \\
    \label{eq:3xn_game_correlations_z}
    \bigg\langle \bigg( \prod_{j \neq k} Z_{A}^{j} \bigg) Z_{B}^{\notqubit{k}} \bigg\rangle &\geq 1 - \varepsilon_0 .
\end{align}
\end{subequations}
The correlations constituting uniformly $\varepsilon_1$-close to perfect local check rounds are, for all distinct $i,j \in \{1,\dots,n\}$,
\begin{subequations}
\label{eq:local_check_correlations}
\begin{align}
    \label{eq:local_check_correlations_x}
    \big\langle X_{A}^{i} X_{B,i}^{i} \big\rangle &\geq 1 - \varepsilon_1 , \\
    \label{eq:local_check_correlations_z}
    \big\langle Z_{A}^{i} Z_{B,j}^{i} \big\rangle &\geq 1 - \varepsilon_1 .
\end{align}
\end{subequations}
The correlations describing uniformly $\varepsilon_2$-close to perfect pair check rounds are, for all distinct $i,j \in \{1,\dots,n\}$,
\begin{subequations}
\label{eq:pair_check_correlations}
\begin{align}
    \label{eq:pair_check_correlations_x}
    \big\langle X_{A}^{i} X_{A}^{j} X_{B}^{i,j} \big\rangle &\geq 1 - \varepsilon_2 , \\
    \label{eq:pair_check_correlations_z}
    \big\langle Z_{A}^{i} Z_{A}^{j} Z_{B}^{i,j} \big\rangle &\geq 1 - \varepsilon_2 .
\end{align}
\end{subequations}
From the assumption that all of these correlations are satisfied for our unknown observables, we will deduce appropriate commutation and anticommutation relations which imply the existence of a local self-testing isometry by \cref{thm:isometry}.

\subsection{Commutation and anticommutation relations}
\label{sec:3xn_commutations}

Here we will deduce the appropriate state-dependent commutation and anticommutation relations of our unknown reflections from which a local self-testing isometry can be constructed.
The results of this section are summarized in the following theorem.
\begin{theorem}
\label{thm:3xn_commutation_summary}
    Let $i,j,k,l \in \{1,\dots,n\}$ be such that $i \neq k$ and $j \neq l$.
    We have correlations between each unknown observable of Alice with each of the corresponding observables on Bob's side
    \begin{align}
        \big\lVert \big( X_{A}^{i} - X_{B,i}^{i} \big) \ket{\Psi} \big\rVert &\leq \sqrt{2\varepsilon_1} , \\
        \big\lVert \big( Z_{A}^{i} - Z_{B,k}^{i} \big) \ket{\Psi} \big\rVert &\leq \sqrt{2\varepsilon_1} .
    \end{align}
    We have the state-dependent anticommutativity of all unknown $X$ observables with all unknown $Z$ observables corresponding to the same qubit
    \begin{align}
        \big\lVert \big\{ X_{A}^{i} , Z_{A}^{i} \big\} \ket{\Psi} \big\rVert
        &\leq 3n \sqrt{2\varepsilon_0} + 2(n-1) \sqrt{2\varepsilon_2}
        + \bigg( \frac{13(n-1)}{2} + 17 \bigg) \sqrt{2\varepsilon_1} , \\
        \big\lVert \big\{ X_{B,i}^{i}, Z_{B,k}^{i} \big\} \ket{\Psi} \big\rVert
        &\leq 3n \sqrt{2\varepsilon_0} + 2(n-1) \sqrt{2\varepsilon_2}
        + \bigg( \frac{13(n-1)}{2} + 21 \bigg) \sqrt{2\varepsilon_1} .
    \end{align}
    Finally, we have the state-dependent commutativity of unknown $X$ and $Z$ observables.
    On Bob's side we have
    \begin{align}
        \big\lVert \big[ X_{B,i}^{i}, X_{B,j}^{j} \big] \ket{\Psi} \big\rVert &\leq 4\sqrt{2\varepsilon_1} , \\
        \big\lVert \big[ Z_{B,k}^{i}, Z_{B,l}^{j} \big] \ket{\Psi} \big\rVert &\leq 4\sqrt{2\varepsilon_1} ;
    \end{align}
    and moreover restricting to observables corresponding to different qubits $i \neq j$
    \begin{equation}
        \big\lVert \big[ X_{B,i}^{i}, Z_{B,l}^{j} \big] \ket{\Psi} \big\rVert \leq 8\sqrt{2\varepsilon_1} .
    \end{equation}
    On Alice's side, for different qubits $i \neq j$, we have
    \begin{equation}
        \big\lVert \big[ M_{A}^{i}, N_{A}^{j} \big] \ket{\Psi} \big\rVert \leq 4\sqrt{2\varepsilon_1} ,
    \end{equation}
    where $M$ and $N$ can be either of $X$ and $Z$.
\end{theorem}
\begin{proof}
    Combine \cref{prop:3xn_correlation_estimation,prop:3xn_commutations,prop:3xn_anticommutations}.
\end{proof}

We begin by expressing the correlations of \cref{eq:local_check_correlations}, between those observables of the players corresponding to local Pauli observables acting on the same qubit, in terms of norms. 
\begin{proposition}[Correlation]
\label{prop:3xn_correlation_estimation}
    For all distinct $i,j \in \{1,\dots,n\}$ we have the correlation estimates
    \begin{subequations}
    \begin{align}
        \big\lVert \big( X_{A}^{i} - X_{B,i}^{i} \big) \ket{\Psi} \big\rVert &\leq \sqrt{2\varepsilon_1} , \\
        \big\lVert \big( Z_{A}^{i} - Z_{B,j}^{i} \big) \ket{\Psi} \big\rVert &\leq \sqrt{2\varepsilon_1} .
    \end{align}
    \end{subequations}
\end{proposition}
\begin{proof}
    Apply \cref{lem:state_estimate} to the correlations given in \cref{eq:local_check_correlations}.
\end{proof}

We now show the required state-dependent commutation relations for observables that correspond to local Pauli observables acting on different qubits.
Since observables of Alice are defined to commute exactly with those of Bob, it is only necessary to consider state-dependent commutation relations on each side separately.
\begin{proposition}[Commutation]
\label{prop:3xn_commutations}
    For all $i,j,k,l \in \{1,\dots,n\}$ such that $i \neq k$ and $j \neq l$ we have
    \begin{subequations}
    \begin{align}
        \label{eq:bob_3xn_comm_xx}
        \big\lVert \big[ X_{B,i}^{i}, X_{B,j}^{j} \big] \ket{\Psi} \big\rVert &\leq 4\sqrt{2\varepsilon_1} , \\
        \label{eq:bob_3xn_comm_zz}
        \big\lVert \big[ Z_{B,k}^{i}, Z_{B,l}^{j} \big] \ket{\Psi} \big\rVert &\leq 4\sqrt{2\varepsilon_1} .
    \end{align}
    \end{subequations}
    Moreover if $i \neq j$ we have commutation relations for Bob
    \begin{equation}
        \label{eq:bob_3xn_comm_xz}
        \big\lVert \big[ X_{B,i}^{i}, Z_{B,l}^{j} \big] \ket{\Psi} \big\rVert \leq 8\sqrt{2\varepsilon_1}
    \end{equation}
    and for Alice
    \begin{equation}
    \label{eq:alice_3xn_comm}
        \big\lVert \big[ M_{A}^{i}, N_{A}^{j} \big] \ket{\Psi} \big\rVert \leq 4\sqrt{2\varepsilon_1} ,
    \end{equation}
    where $M$ and $N$ can be either of $X$ and $Z$.
\end{proposition}
\begin{proof}
    As the proof of \cref{prop:commutations}, but using the correlations of \cref{eq:local_check_correlations} instead of \cref{eq:check_correlations}.
\end{proof}

The following proposition states the robust state-dependent anticommutation relations between each pair of unknown $X$ and $Z$ observables corresponding to the same qubit, depending on the correlation errors $\varepsilon_0$, $\varepsilon_1$, and $\varepsilon_2$.
A sketch proof is given below for the ideal case with vanishing errors, with the more lengthy, full proof being the contents of \cref{sec:robust_3xn_anticommutations}.
\begin{restatable}[Anticommutation]{proposition}{anticomm}
\label{prop:3xn_anticommutations}
    For all $i \in \{1,\dots,n\}$ we have state-dependent anticommutation relations for unknown observables of Alice
    \begin{equation}
    \label{eq:3xn_alice_anticomm}
        \big\lVert \big\{ X_{A}^{i}, Z_{A}^{i} \big\} \ket{\Psi} \big\rVert
        \leq 3n \sqrt{2\varepsilon_0} + 2(n-1) \sqrt{2\varepsilon_2}
        + \bigg( \frac{13(n-1)}{2} + 17 \bigg) \sqrt{2\varepsilon_1} .
    \end{equation}
    Furthermore, for all $j \in \{1,\dots,n\}$ distinct from $i$ we have state-dependent anticommutation relations for Bob's check-round observables
    \begin{equation}
    \label{eq:bob_local_check_anticomm}
        \big\lVert \big\{ X_{B,i}^{i}, Z_{B,j}^{i} \big\} \ket{\Psi} \big\rVert
        \leq 3n \sqrt{2\varepsilon_0} + 2(n-1) \sqrt{2\varepsilon_2}
        + \bigg( \frac{13(n-1)}{2} + 21 \bigg) \sqrt{2\varepsilon_1} .
    \end{equation}
\end{restatable}
\begin{proof}[Sketch proof]
    For the sake of sketching the proof, take correlation errors to vanish $\varepsilon_0 = \varepsilon_1 = \varepsilon_2 = 0$.
    We will show the state-dependent anticommutation relation $\big\{ X_{A}^{1}, Z_{A}^{1} \big\} \ket{\Psi} = 0$.
    The relations for observables corresponding to the other qubits follow similarly.

    From the game correlations \cref{eq:3xn_game_correlations_y} we have
    \begin{equation}
        \Bigg( \prod_{k=2}^{n} Z_{B}^{\notqubit{k}} X_{B}^{\notqubit{k}} \Bigg) \ket{\Psi}
        + Z_{B}^{\notqubit{1}} X_{B}^{\notqubit{1}} \ket{\Psi} = 0 ,
    \end{equation}
    where the sign of the first term uses that $n$ is odd.
    Swapping to Alice's side those observables acting immediately on the state and multiplying on the left by appropriate unitary operators gives
    \begin{equation}
        X_{B}^{\notqubit{2}} \Bigg( \prod_{k=3}^{n-1} Z_{B}^{\notqubit{k}} X_{B}^{\notqubit{k}} \Bigg) Z_{B}^{\notqubit{n}} \ket{\Psi}
        + Z_{B}^{\notqubit{2}} Z_{B}^{\notqubit{1}} X_{A}^{n} X_{A}^{1} \ket{\Psi} = 0 .
    \end{equation}

    Rewriting this by commuting those $X$ and $Z$ observables within each term of the product with $k$ odd results in
    \begin{equation}
        \Bigg( \prod_{k=1}^{(n-3)/2} X_{B}^{\notqubit{2k}} X_{B}^{\notqubit{2k+1}} Z_{B}^{\notqubit{2k+1}} Z_{B}^{\notqubit{2k+2}} \Bigg) X_{B}^{\notqubit{n-1}} Z_{B}^{\notqubit{n}} \ket{\Psi}
        + Z_{B}^{\notqubit{2}} Z_{B}^{\notqubit{1}} X_{A}^{n} X_{A}^{1} \ket{\Psi} = 0 .
    \end{equation}
    Using the correlations of \cref{eq:3xn_game_correlations_x,eq:3xn_game_correlations_z} to swap Bob's observables to Alice's side (and freely inserting the identity operator as $X_{A}^{n-1} X_{A}^{n-1}$ into the resulting first term) yields
    \begin{equation}
    \label{eq:alice_swap_step_sketch}
        \Bigg( \prod_{k \neq n} Z_{A}^{k} \Bigg)
        \Bigg( \prod_{k} X_{A}^{k} \Bigg)
        \Bigg( \prod_{k=1}^{(n-3)/2} X_{A}^{n-2k+1} Z_{A}^{n-2k+1} Z_{A}^{n-2k} X_{A}^{n-2k} \Bigg)
        X_{A}^{2} \ket{\Psi}
        + X_{A}^{n} X_{A}^{1} Z_{A}^{1} Z_{A}^{2} \ket{\Psi} = 0 .
    \end{equation}
    From the correlations of \cref{eq:local_check_correlations} we have
    \begin{equation}
    \label{eq:alice_swap_step_estimate_sketch}
        X_{A}^{2} Z_{B,n}^{1} Z_{B,n}^{2} X_{B,n}^{n} X_{B,1}^{1} \ket{\Psi}
        = X_{B,n}^{n} Z_{B}^{1,2} X_{B}^{1,2} \ket{\Psi} .
    \end{equation}
    Hence multiplying \cref{eq:alice_swap_step_sketch} on the left by $Z_{B,n}^{1} Z_{B,n}^{2} X_{B,n}^{n} X_{B,1}^{1}$, applying \cref{eq:alice_swap_step_estimate_sketch} via the triangle inequality in its first term (commuting the resulting observables for Bob with the existing observables of Alice), and in its second term using the correlations of \cref{eq:local_check_correlations},
    \begin{multline}
    \label{eq:pair_estimate_step_sketch}
        X_{B,n}^{n} Z_{B}^{1,2} X_{B}^{1,2}
        \Bigg( \prod_{k \neq n} Z_{A}^{k} \Bigg)
        \Bigg( \prod_{k} X_{A}^{k} \Bigg)
        \Bigg( \prod_{k=1}^{(n-3)/2} X_{A}^{n-2k+1} Z_{A}^{n-2k+1} Z_{A}^{n-2k} X_{A}^{n-2k} \Bigg) \ket{\Psi} \\
        + \big( X_{A}^{n} X_{A}^{1} Z_{A}^{1} Z_{A}^{2} \big)^2 \ket{\Psi} = 0 .
    \end{multline}

    \Cref{lem:product_to_pair_obs} shows for all $k \in \big\{ 1, \dots, \frac{n-3}{4} \big\}$ that in particular
    \begin{equation}
    \label{eq:product_to_pair_obs_sketch}
    \begin{split}
        \big( & X_{A}^{4k+2} Z_{A}^{4k+2} Z_{A}^{4k+1} X_{A}^{4k+1} \big)
        \big( X_{A}^{4k} Z_{A}^{4k} Z_{A}^{4k-1} X_{A}^{4k-1} \big) \ket{\Psi} \\
        &= X_{B}^{4k-1,4k+1} Z_{B}^{4k-1,4k+1} Z_{B}^{4k,4k+2} X_{B}^{4k,4k+2} \ket{\Psi} .
    \end{split}
    \end{equation}
    Since $n \equiv 3 \pmod{4}$, we can consider successive pairs of terms in the final product of \cref{eq:pair_estimate_step_sketch}.
    We can replace each pair of these terms using pair check observables by repeatedly applying \cref{eq:product_to_pair_obs_sketch} and commuting the resulting observables of Bob with those of Alice.
    This gives
    \begin{multline}
    \label{eq:paired_step_sketch}
        X_{B,n}^{n} Z_{B}^{1,2} X_{B}^{1,2}
        \Bigg( \prod_{k=1}^{(n-3)/4} X_{B}^{4k-1,4k+1} Z_{B}^{4k-1,4k+1} Z_{B}^{4k,4k+2} X_{B}^{4k,4k+2} \Bigg)
        \Bigg( \prod_{k \neq n} Z_{A}^{k} \Bigg)
        \Bigg( \prod_{k} X_{A}^{k} \Bigg) \ket{\Psi} \\
        + \big( X_{A}^{n} X_{A}^{1} Z_{A}^{1} Z_{A}^{2} \big)^2 \ket{\Psi} = 0 .
    \end{multline}
    \Cref{lem:alice_to_pairs} with $\sigma = \operatorname{id}$ chosen to be the identity permutation shows
    \begin{multline}
    \label{eq:alice_to_pairs_sketch}
        \Bigg( \prod_{k \neq n} Z_{A}^{k} \Bigg)
        \Bigg( \prod_{k} X_{A}^{k} \Bigg) \ket{\Psi} \\
        = X_{A}^{n} X_{B}^{1,2} \Bigg( \prod_{k=1}^{(n-3)/4} X_{B}^{4k-1,4k+1} X_{B}^{4k,4k+2} \Bigg)
        Z_{B}^{1,2} \Bigg( \prod_{k=1}^{(n-3)/4} Z_{B}^{4k-1,4k+1} Z_{B}^{4k,4k+2} \Bigg) \ket{\Psi} .
    \end{multline}

    If we assume that all pair check observables appearing in \cref{eq:paired_step_sketch} are measured as part of the same (pair check round) input for Bob (which is compatible with an honest strategy since all of these observables have either disjoint or identical superscript index pairs to all others), then all such observables mutually commute.
    Thus applying \cref{eq:alice_to_pairs_sketch} to \cref{eq:paired_step_sketch} and using the involutory property of all pair check observables to achieve many cancellations, we get
    \begin{equation}
    \label{eq:cancelled_sketch}
        X_{A}^{n} X_{B,n}^{n} \ket{\Psi}
        + \big( X_{A}^{n} X_{A}^{1} Z_{A}^{1} Z_{A}^{2} \big)^2 \ket{\Psi} = 0 .
    \end{equation}
    It should be noted that, for the simplicity of this sketch, the set of mutually commuting pair check observables used as an input here does not necessarily match one of the inputs defined in \cref{eq:pair_check_observables}.
    Nonetheless, it is still the case that only $n$ such sets must be used to complete the proof for all anticommutation relations of Alice's observables, and (with the proof essentially unchanged) the set used here matches one of those in \cref{eq:pair_check_observables} under a suitable permutation of the qubit labels.

    Applying the correlations of \cref{eq:local_check_correlations_x} once in the first term of \cref{eq:cancelled_sketch} and then multiplying on the left by $Z_{A}^{2} Z_{A}^{1} X_{A}^{1} X_{A}^{n}$ gives
    \begin{equation}
    \label{eq:3xn_alice_pair_anticomm_sketch}
        \big\{ X_{A}^{1} X_{A}^{n}, Z_{A}^{1} Z_{A}^{2} \big\} \ket{\Psi} = 0 .
    \end{equation}
    By identical argument to the proof of \cref{prop:anticommutations}, but using \cref{prop:3xn_correlation_estimation,prop:3xn_commutations} instead of \cref{prop:correlation_estimation,prop:commutations} and using \cref{eq:3xn_alice_pair_anticomm_sketch} in place of \cref{lem:alice_pair_anticomm}, this implies the desired state-dependent anticommutation relation $\big\{ X_{A}^{1}, Z_{A}^{1} \big\} \ket{\Psi} = 0$ for Alice's observables.

    The state-dependent anticommutation relations for Bob can all be obtained by simple application of \cref{prop:3xn_correlation_estimation}, given those just proved for Alice's observables.
\end{proof}

\section{Discussion}
\label{sec:discussion}

In this work we introduced one-side-local quantum strategies for the magic square and \rectdim{3}{n} magic rectangle games that win with certainty.
We then supplemented these strategies with some extra correlations obtained via ``check'' rounds to obtain the desired self-tests.
Our final result is a parallel self-test for $n$ maximally entangled Bell states which has several practical advantages over other protocols.
Being a parallel self-test of $n$ Bell states, our protocol makes no assumptions within the $n$ single-qubit systems of each side.

We examine first the experimental requirements of realizing our self-test---something that is determined by the honest runs.
All observables used in the honest strategy for our self-test can be implemented as the tensor product of at most two Pauli operators (of the same type) acting on different pairs of qubits.
A unique advantage of our work is that, moreover, Alice need only ever make local measurements of single-qubit Pauli observables in the honest case.
This is especially important for major uses of self-testing.
For example, in the context of delegated quantum computation the ``client'' could have very limited quantum capabilities.
It suffices that they are able to measure single qubits in Pauli bases.

Another interesting property of our self-test concerns its communication complexity.
Of particular importance is the size of input questions, which quantify how much randomness must be consumed by the protocol in each round of interaction.
Our test requires constant size (\SI{1}{trit}) input questions for Alice, and for Bob requires $\bigO(\log{n}) \, \si{\bit}$ inputs.
With a few exceptions \cite{sarkar2021self,supic2021device,fu2022constant,manvcinska2021constant} (in each of which robustness is either not explicitly constructed or doubly exponential in $n$), other works have achieved at best logarithmic input complexities (see for example \cite{chao2018test,natarajan2018low}).
In our protocol, one of the players need only receive questions of a constant size.
Players must each output $\bigO(n) \, \si{\bit}$ answers, except for in game rounds, in which Bob need only return \SI{2}{\bit} outputs.

Our protocol also has the practical advantage that it makes use of solely perfect correlations; any optimal strategy succeeds with certainty, thus requiring fewer rounds of experiment to achieve a desired statistical confidence.

The final figure of merit that we consider is robustness to noise.
Given correlations that are at worst $\varepsilon$-close to perfect, using a self-testing theorem that can be found in \cite{coladangelo2017parallel}, our results achieve a robustness that is $\bigO \big( n^{\frac{5}{2}} \sqrt{\varepsilon} \big)$ for the collection of Bell states and all single-qubit Pauli observables.
That is, to achieve a robustness $\delta$ it is sufficient that $\varepsilon(n, \delta) \in \bigO(n^{-5} \delta^{2})$.
The self-testing works of \textcite{coladangelo2017parallel} and \textcite{coudron2016parallel} using instead the parallel repetition of the magic square game as a basis perform slightly better in this regard, with $\varepsilon(n, \delta) \in \bigO(n^{-3} \delta^{2})$ and $\varepsilon(n, \delta) \in \bigO(n^{-4} \delta^{4})$ being sufficient for robustness $\delta$, respectively.
The work of \textcite{coudron2016parallel} achieves robustness for observables acting on all qubits simultaneously, however, both works are examples of strictly parallel self-tests and thus necessarily require $\bigO(n) \, \si{\bit}$ inputs.
A protocol of \textcite{natarajan2017quantum} exhibits the interesting property that its robustness does not depend on $n$.
The same authors later extended this work to have communication complexity only logarithmic in the number of entangled states to be certified.
The protocol, however, instead self-tests $N$ maximally entangled \emph{qudit} states and corresponding single-\emph{qudit} Pauli observables defined over a finite field $\mathbb{F}_{q}$, where $q$ increases with $N$ \cite{natarajan2018low}.
It is unclear whether the honest strategy provided can be realized with local measurements with respect to Bell states.\footnote{While maximally-entangled qudit states and generalized qudit Pauli measurement projectors are isomorphic to tensor products of $\ket{\Phi^{+}}$ and \emph{qubit} Pauli measurement projectors respectively (as shown in a lemma of \cite{ji2021mip}), it is not clear that all of the measurements used in the qudit honest strategy of \cite{natarajan2018low} can be mapped under such an isomorphism to local measurements with respect to each two-dimensional register (in general they may become entangled measurements).}

Our protocol is unique in that it achieves several desirable properties simultaneously.
The prover with minimal quantum-technological capabilities (the client) need only make local single-qubit measurements in Pauli bases upon accepting questions all of constant size.
Despite this, our protocol relies entirely on perfect correlations, maintains a noise tolerance comparable with that of most others, and requires questions provided to the server to be of size at most logarithmic in the number of Bell states tested.
Sample comparisons with some other protocols can be found in \cref{tab:protocol_comparison}.
The list of works included is not exhaustive, and other figures of merit could also be considered depending on intended applications.

\begin{table*}
    \begin{tabular}{lcccccc}
        \toprule
        Protocol & Local & Pauli & Perf. corr. & Error tol. & \multicolumn{2}{c}{Input size} \\
        \cmidrule(l){6-7}
        & & & & $\varepsilon(n, \delta)$ & Alice & Bob \\
        \midrule
        \rectdim{3}{n} protocol (\textbf{this work}) & Alice\cellcolor{yellow} & Yes\cellcolor{green} & Yes\cellcolor{green} & $\bigO(n^{-5} \delta^{2})$\cellcolor{yellow} & $\bigO(1)$\cellcolor{green} & $\bigO(\log{n})$\cellcolor{yellow} \\
        \textcite{supic2021device} & \multicolumn{3}{c}{Depends on base self-tests} & N/A & \multicolumn{2}{c}{$\bigO(1)$\cellcolor{green}} \\
        \textcite{chao2018test} & Yes\cellcolor{green} & No\cellcolor{red} & No\cellcolor{red} & $\bigO(n^{-5} \delta^{2})$\cellcolor{yellow} & \multicolumn{2}{c}{$\bigO(\log{n})$\cellcolor{yellow}} \\
        \textcite{natarajan2018low} & No\cellcolor{red} & No\cellcolor{red} & Yes\cellcolor{green} & $\bigO(\poly(\delta))$\cellcolor{green} & \multicolumn{2}{c}{$\bigO(\poly(\log{n}))$\cellcolor{yellow}} \\
        \textcite{natarajan2017quantum} & \multicolumn{3}{c}{As CHSH or magic square} & $\bigO(\delta^{16})$\cellcolor{green} & \multicolumn{2}{c}{$\bigO(n)$\cellcolor{red}} \\
        \textcite{coladangelo2017parallel} (magic square) & No\cellcolor{red} & Yes\cellcolor{green} & Yes\cellcolor{green} & $\bigO(n^{-3} \delta^{2})$\cellcolor{yellow} & \multicolumn{2}{c}{$\bigO(n)$\cellcolor{red}} \\
        \textcite{coladangelo2017parallel} (CHSH) & Yes\cellcolor{green} & No\cellcolor{red} & No\cellcolor{red} & $\bigO(n^{-3} \delta^{2})$\cellcolor{yellow} & \multicolumn{2}{c}{$\bigO(n)$\cellcolor{red}} \\
        \textcite{coudron2016parallel} & No\cellcolor{red} & Yes\cellcolor{green} & Yes\cellcolor{green} & $\bigO(n^{-4} \delta^{4})$\cellcolor{yellow} & \multicolumn{2}{c}{$\bigO(n)$\cellcolor{red}} \\
        \textcite{mckague2016self} (Mayers--Yao) & Yes\cellcolor{green} & No\cellcolor{red} & No\cellcolor{red} & $\bigO(n^{-8} \delta^{8})$\cellcolor{yellow} & \multicolumn{2}{c}{$\bigO(\log{\log{n}})$\cellcolor{yellow}} \\
        \bottomrule
    \end{tabular}
    \caption{
        Comparison between certain protocols capable of self-testing $n$ EPR pairs in parallel.
        Cells highlighted in green depict favorable comparisons within the property being considered.
        Those in red compare unfavorably and those in yellow neutrally.
        Whether the honest strategy of each protocol uses only local (single-qubit) measurements, is constructed entirely from measurements of the Pauli group (on standard Bell states), and makes use of only perfect correlations (so that the strategy wins with certainty) are considered.
        The error tolerance $\varepsilon(n ,\delta)$ is a sufficient maximum error in the observed correlations so that the states and measurements tested (up to local isometry) are a distance at most $\delta$ from ideal.
        Input question sizes (the amounts of randomness consumed) are given in units of information.
    }
    \label{tab:protocol_comparison}
\end{table*}

\subparagraph{Future works.}
Aside from our self-testing result, all self-tests whose honest strategies rely solely on the magic square game (such as those of \cite{coladangelo2017parallel,coudron2016parallel}) can of course be implemented using our one-side-local strategy if desired.
It may also be possible to use our one-side-local strategy as a direct replacement for honest subroutines in other protocols (such as the CHSH game in the protocol of \textcite{chao2018test} or for the anticommutation test of \textcite{natarajan2017quantum}), allowing them to function with the additional benefits of local Pauli measurements and perfect correlations at the same time.

In this work, we made use of a theorem of \textcite{coladangelo2017parallel} to translate our main state-dependent commutation/anticommutation results into a proper self-testing statement on the existence of a desired local isometry.
Other choices of isometry could equally-well have been made.
On the one hand, other results based on the relevant commutativity and anticommutativity of untrusted observables exist.
For example, a result of \textcite[Lemma~6]{mckague2016self} could be directly substituted for that used here, offering the additional property of simultaneously testing multiple Pauli measurements at the cost of poorer robustness scaling $\varepsilon(n, \delta) \in \bigO(n^{-8} \delta^{4})$.
On the other hand, it would be interesting to examine the plausibility of more robust isometries for our self-test.
Such isometries could arise either as improved general techniques for the construction of self-testing isometries given certain relations between the untrusted observables (similar to \cite{mckague2016self,coladangelo2017parallel}), or alternatively in the form of specially-constructed isometries making use of features unique to the testing scenario.
Another possible future direction is to study the robustness of our protocol experimentally (or numerically under the semidefinite-programming characterization of quantum correlations \cite{navascues2007bounding,navascues2008convergent,yang2014robust,bancal2015physical}).

Adaptation of our results for device-independent versions of delegated verifiable blind quantum computation protocols, or other secure quantum computation protocols \cite{kashefi2017garbled,kashefi2017multiparty}, could be explored.
The utility of our protocol for device-independent quantum key distribution could also be examined.

\section*{Acknowledgments}

S.A.A. gratefully acknowledges EPSRC studentship funding under grant number EP/R513209/1.
P.W. acknowledges support by the UK Hub in Quantum Computing and Simulation, part of the UK National Quantum Technologies Programme with funding from EPSRC grant EP/T001062/1.

\appendix

\section{Robust anticommutation relations}
\label{sec:robust_3xn_anticommutations}

\begin{lemma}
\label{lem:product_to_pair_obs}
    For all distinct $i,j,k,l \in \{1,\dots,n\}$ we have the estimate between Alice's observables and Bob's pair check observables,
    \begin{equation}
        \big\lVert
        \big( X_{A}^{l} Z_{A}^{l} Z_{A}^{k} X_{A}^{k} \big)
        \big( X_{A}^{j} Z_{A}^{j} Z_{A}^{i} X_{A}^{i} \big) \ket{\Psi}
        - X_{B}^{i,k} Z_{B}^{i,k} Z_{B}^{j,l} X_{B}^{j,l} \ket{\Psi}
        \big\rVert
        \leq 18 \sqrt{2\varepsilon_1} + 4 \sqrt{2\varepsilon_2} .
    \end{equation}
\end{lemma}
\begin{proof}
    First, commuting $X_{A}^{j}$ with $X_{A}^{k}$ (as they correspond to the same input) and then using \cref{prop:3xn_correlation_estimation} and the triangle inequality to swap four of Alice's observables to Bob's side, we have
    \begin{equation}
        \big\lVert
        \big( X_{A}^{l} Z_{A}^{l} Z_{A}^{k} X_{A}^{k} \big)
        \big( X_{A}^{j} Z_{A}^{j} Z_{A}^{i} X_{A}^{i} \big) \ket{\Psi}
        - X_{A}^{l} Z_{A}^{l} Z_{A}^{k} X_{A}^{j} X_{B,i}^{i} Z_{B,k}^{i} Z_{B,k}^{j} X_{B,k}^{k} \ket{\Psi}
        \big\rVert
        \leq 4 \sqrt{2\varepsilon_1} .
    \end{equation}
    Commuting $X_{B,k}^{k}$ with observables of the same input ($Z_{B,k}^{i}$ and $Z_{B,k}^{j}$) and then again using the correlations to swap observables back to Alice's side gives
    \begin{equation}
        \big\lVert
        \big( X_{A}^{l} Z_{A}^{l} Z_{A}^{k} X_{A}^{k} \big)
        \big( X_{A}^{j} Z_{A}^{j} Z_{A}^{i} X_{A}^{i} \big) \ket{\Psi}
        - X_{A}^{l} Z_{A}^{l} Z_{A}^{k} X_{A}^{j} Z_{A}^{j} Z_{A}^{i} X_{A}^{k} X_{A}^{i} \ket{\Psi}
        \big\rVert
        \leq 8 \sqrt{2\varepsilon_1} .
    \end{equation}
    Applying \cref{eq:pair_check_correlations_x} correlations between Alice's observables and Bob's pair check observables once followed by swapping five of Alice's observables to Bob's side gives
    \begin{equation}
        \big\lVert
        \big( X_{A}^{l} Z_{A}^{l} Z_{A}^{k} X_{A}^{k} \big)
        \big( X_{A}^{j} Z_{A}^{j} Z_{A}^{i} X_{A}^{i} \big) \ket{\Psi}
        - X_{B}^{i,k} X_{A}^{l} Z_{B,l}^{i} Z_{B,l}^{j} X_{B,j}^{j} Z_{B,j}^{k} Z_{B,j}^{l} \ket{\Psi}
        \big\rVert
        \leq 13 \sqrt{2\varepsilon_1} + \sqrt{2\varepsilon_2} .
    \end{equation}
    Commuting $X_{B,j}^{j}$ with observables of the same input ($Z_{B,j}^{k}$ and $Z_{B,j}^{l}$) and again swapping local check observables back to Alice's side yields
    \begin{equation}
        \big\lVert
        \big( X_{A}^{l} Z_{A}^{l} Z_{A}^{k} X_{A}^{k} \big)
        \big( X_{A}^{j} Z_{A}^{j} Z_{A}^{i} X_{A}^{i} \big) \ket{\Psi}
        - X_{B}^{i,k} X_{A}^{l} X_{A}^{j} Z_{A}^{l} Z_{A}^{k} Z_{A}^{j} Z_{A}^{i} \ket{\Psi}
        \big\rVert
        \leq 18 \sqrt{2\varepsilon_1} + \sqrt{2\varepsilon_2} .
    \end{equation}
    Finally, commuting $Z_{A}^{j}$ with $Z_{A}^{k}$ and applying three correlations of \cref{eq:pair_check_correlations} to switch all observables of Alice with pair check observables of Bob gives the result.
\end{proof}

\begin{lemma}
\label{lem:alice_to_pairs}
    For any permutation $\sigma$ of $\{1,\dots,n\}$, letting $\sigma_{k} = \sigma(k)$ for each $k$, we have the estimate
    \begin{multline}
    \label{eq:alice_to_pairs}
        \Bigg\lVert
        \Bigg( \prod_{k \neq n} Z_{A}^{\sigma_k} \Bigg)
        \Bigg( \prod_{k} X_{A}^{\sigma_k} \Bigg) \ket{\Psi}
        - X_{A}^{\sigma_n} X_{B}^{\sigma_1, \sigma_2} \Bigg( \prod_{k=1}^{(n-3)/4} X_{B}^{\sigma_{4k-1}, \sigma_{4k+1}} X_{B}^{\sigma_{4k}, \sigma_{4k+2}} \Bigg) \\
        Z_{B}^{\sigma_1, \sigma_2} \Bigg( \prod_{k=1}^{(n-3)/4} Z_{B}^{\sigma_{4k-1}, \sigma_{4k+1}} Z_{B}^{\sigma_{4k}, \sigma_{4k+2}} \Bigg) \ket{\Psi}
        \Bigg\rVert
        \leq 2n \sqrt{2\varepsilon_1} + (n-1) \sqrt{2\varepsilon_2} .
    \end{multline}
\end{lemma}
\begin{proof}
    Noting that all the $X_{A}^{k}$ pairwise commute and using the correlations of \cref{eq:pair_check_correlations_x} to swap Alice's observables with Bob's pair check observables,
    \begin{multline}
    \label{eq:x_swapped}
        \Bigg\lVert
        \Bigg( \prod_{k \neq n} Z_{A}^{\sigma_k} \Bigg)
        \Bigg( \prod_{k} X_{A}^{\sigma_k} \Bigg) \ket{\Psi}
        - X_{B}^{\sigma_1, \sigma_2} \Bigg( \prod_{k=1}^{(n-3)/4} X_{B}^{\sigma_{4k-1}, \sigma_{4k+1}} X_{B}^{\sigma_{4k}, \sigma_{4k+2}} \Bigg)
        \Bigg( \prod_{k \neq n} Z_{A}^{\sigma_k} \Bigg) X_{A}^{\sigma_n} \ket{\Psi}
        \Bigg\rVert \\
        \leq \frac{n-1}{2} \sqrt{2\varepsilon_2} .
    \end{multline}
    Consider only the final part of the second term in \cref{eq:x_swapped}.
    We can repeatedly apply the triangle inequality with \cref{prop:3xn_correlation_estimation} to write
    \begin{equation}
        \Bigg\lVert
        \Bigg( \prod_{k \neq n} Z_{A}^{\sigma_k} \Bigg) X_{A}^{\sigma_n} \ket{\Psi}
        - X_{B,\sigma_n}^{\sigma_n} \Bigg( \prod_{k \neq n} Z_{B,\sigma_n}^{\sigma_k} \Bigg) \ket{\Psi}
        \Bigg\rVert
        \leq n \sqrt{2\varepsilon_1} .
    \end{equation}
    Since all of Bob's observables in this equation correspond to the same input, we can commute $X_{B,\sigma_n}^{\sigma_n}$ with the product to its right and then use \cref{prop:3xn_correlation_estimation} again to give
    \begin{equation}
        \Bigg\lVert
        \Bigg( \prod_{k \neq n} Z_{A}^{\sigma_k} \Bigg) X_{A}^{\sigma_n} \ket{\Psi}
        - X_{A}^{\sigma_n} \Bigg( \prod_{k \neq n} Z_{A}^{\sigma_k} \Bigg) \ket{\Psi}
        \Bigg\rVert
        \leq 2n \sqrt{2\varepsilon_1} .
    \end{equation}
    Combining this with \cref{eq:x_swapped} via the triangle inequality yields
    \begin{multline}
    \label{eq:local_x_commuted}
        \Bigg\lVert
        \Bigg( \prod_{k \neq n} Z_{A}^{\sigma_k} \Bigg)
        \Bigg( \prod_{k} X_{A}^{\sigma_k} \Bigg) \ket{\Psi}
        - X_{A}^{\sigma_n} X_{B}^{\sigma_1, \sigma_2} \Bigg( \prod_{k=1}^{(n-3)/4} X_{B}^{\sigma_{4k-1}, \sigma_{4k+1}} X_{B}^{\sigma_{4k}, \sigma_{4k+2}} \Bigg)
        \Bigg( \prod_{k \neq n} Z_{A}^{\sigma_k} \Bigg) \ket{\Psi}
        \Bigg\rVert \\
        \leq 2n \sqrt{2\varepsilon_1} + \frac{n-1}{2} \sqrt{2\varepsilon_2} .
    \end{multline}
    Finally, since all the $Z_{A}^{k}$ pairwise commute, the correlations of \cref{eq:pair_check_correlations_z} imply
    \begin{equation}
        \Bigg\lVert
        \Bigg( \prod_{k \neq n} Z_{A}^{\sigma_k} \Bigg) \ket{\Psi}
        - Z_{B}^{\sigma_1, \sigma_2} \Bigg( \prod_{k=1}^{(n-3)/4} Z_{B}^{\sigma_{4k-1}, \sigma_{4k+1}} Z_{B}^{\sigma_{4k}, \sigma_{4k+2}} \Bigg) \ket{\Psi}
        \Bigg\rVert
        \leq \frac{n-1}{2} \sqrt{2\varepsilon_2} .
    \end{equation}
    Combining this with the previous \cref{eq:local_x_commuted} using the triangle inequality yields the result.
\end{proof}

We now exhibit the full proof of \cref{prop:3xn_anticommutations} with nonzero correlation errors.
\anticomm*
\begin{proof}
    Let $i \in \{1,\dots,n\}$ and let $\sigma_k = \sigma(k)$ for each $k \in \{1,\dots,n\}$, where $\sigma$ is some permutation of $\{1,\dots,n\}$.
    Assume that $\sigma$ is such that $\sigma_1 = i$.
    From the game correlations \cref{eq:3xn_game_correlations_y} we have
    \begin{equation}
        \Bigg\lVert \Bigg( \prod_{k=2}^{n} Y_{A}^{\sigma_k} \Bigg) \ket{\Psi}
        + Z_{B}^{\notqubit{\sigma_1}} X_{B}^{\notqubit{\sigma_1}} \ket{\Psi} \Bigg\rVert
        \leq \sqrt{2\varepsilon_0} .
    \end{equation}
    Again, from the same correlations,
    \begin{equation}
        \Bigg\lVert \Bigg( \prod_{k=2}^{n} Z_{B}^{\notqubit{\sigma_k}} X_{B}^{\notqubit{\sigma_k}} \Bigg) \ket{\Psi}
        + Z_{B}^{\notqubit{\sigma_1}} X_{B}^{\notqubit{\sigma_1}} \ket{\Psi} \Bigg\rVert
        \leq n \sqrt{2\varepsilon_0} ,
    \end{equation}
    where the sign of the first term uses that $n$ is odd.
    Now using the game correlations \cref{eq:3xn_game_correlations_x},
    \begin{equation}
        \Bigg\lVert \Bigg( \prod_{k=2}^{n-1} Z_{B}^{\notqubit{\sigma_k}} X_{B}^{\notqubit{\sigma_k}} \Bigg) Z_{B}^{\notqubit{\sigma_n}} \Bigg( \prod_{k \neq n} X_{A}^{\sigma_k} \Bigg) \ket{\Psi}
        + Z_{B}^{\notqubit{\sigma_1}} \Bigg( \prod_{k \neq 1} X_{A}^{\sigma_k} \Bigg) \ket{\Psi} \Bigg\rVert
        \leq (n+2) \sqrt{2\varepsilon_0} .
    \end{equation}
    Multiplying on the left by the unitary operators $\prod_{k \neq n} X_{A}^{\sigma_k}$ and $Z_{B}^{\notqubit{\sigma_2}}$ leaves the norm unchanged and gives
    \begin{equation}
        \Bigg\lVert X_{B}^{\notqubit{\sigma_2}} \Bigg( \prod_{k=3}^{n-1} Z_{B}^{\notqubit{\sigma_k}} X_{B}^{\notqubit{\sigma_k}} \Bigg) Z_{B}^{\notqubit{\sigma_n}} \ket{\Psi}
        + Z_{B}^{\notqubit{\sigma_2}} Z_{B}^{\notqubit{\sigma_1}} X_{A}^{\sigma_n} X_{A}^{\sigma_1} \ket{\Psi} \Bigg\rVert
        \leq (n+2) \sqrt{2\varepsilon_0} .
    \end{equation}
    Rewriting this by commuting those $X$ and $Z$ observables within each term of the product with $k$ odd results in
    \begin{equation}
        \Bigg\lVert \Bigg( \prod_{k=1}^{(n-3)/2} X_{B}^{\notqubit{\sigma_{2k}}} X_{B}^{\notqubit{\sigma_{2k+1}}} Z_{B}^{\notqubit{\sigma_{2k+1}}} Z_{B}^{\notqubit{\sigma_{2k+2}}} \Bigg) X_{B}^{\notqubit{\sigma_{n-1}}} Z_{B}^{\notqubit{\sigma_n}} \ket{\Psi}
        + Z_{B}^{\notqubit{\sigma_2}} Z_{B}^{\notqubit{\sigma_1}} X_{A}^{\sigma_n} X_{A}^{\sigma_1} \ket{\Psi} \Bigg\rVert
        \leq (n+2) \sqrt{2\varepsilon_0} .
    \end{equation}
    Using the correlations of \cref{eq:3xn_game_correlations_x,eq:3xn_game_correlations_z} to swap Bob's observables to Alice's side (and freely inserting the identity operator as $X_{A}^{\sigma_{n-1}} X_{A}^{\sigma_{n-1}}$ into the resulting first term) yields
    \begin{multline}
    \label{eq:alice_swap_step}
        \Bigg\lVert
        \Bigg( \prod_{k \neq n} Z_{A}^{\sigma_k} \Bigg)
        \Bigg( \prod_{k} X_{A}^{\sigma_k} \Bigg)
        \Bigg( \prod_{k=1}^{(n-3)/2} X_{A}^{\sigma_{n-2k+1}} Z_{A}^{\sigma_{n-2k+1}} Z_{A}^{\sigma_{n-2k}} X_{A}^{\sigma_{n-2k}} \Bigg)
        X_{A}^{\sigma_2} \ket{\Psi} \\
        + X_{A}^{\sigma_n} X_{A}^{\sigma_1} Z_{A}^{\sigma_1} Z_{A}^{\sigma_2} \ket{\Psi}
        \Bigg\rVert
        \leq 3n \sqrt{2\varepsilon_0} .
    \end{multline}
    Now notice from the correlations of \cref{eq:local_check_correlations} we have the estimate
    \begin{equation}
    \label{eq:alice_swap_step_estimate}
        \big\lVert
        X_{A}^{\sigma_2} Z_{B,\sigma_n}^{\sigma_1} Z_{B,\sigma_n}^{\sigma_2} X_{B,\sigma_n}^{\sigma_n} X_{B,\sigma_1}^{\sigma_1} \ket{\Psi}
        - X_{B,\sigma_n}^{\sigma_n} Z_{B}^{\sigma_1, \sigma_2} X_{B}^{\sigma_1, \sigma_2} \ket{\Psi}
        \big\rVert
        \leq 3 \sqrt{2\varepsilon_1} + 2 \sqrt{2\varepsilon_2} ,
    \end{equation}
    where we achieved this by commuting $X_{B,\sigma_n}^{\sigma_n}$ with other observables of the same input and converting local check observables to observables of Alice and then to pair check observables.
    Hence multiplying \cref{eq:alice_swap_step} on the left by $Z_{B,\sigma_n}^{\sigma_1} Z_{B,\sigma_n}^{\sigma_2} X_{B,\sigma_n}^{\sigma_n} X_{B,\sigma_1}^{\sigma_1}$, applying \cref{eq:alice_swap_step_estimate} via the triangle inequality in its first term (commuting the resulting observables for Bob with the existing observables of Alice), and in its second term using the correlations of \cref{eq:local_check_correlations},
    \begin{multline}
    \label{eq:pair_estimate_step}
        \Bigg\lVert
        X_{B,\sigma_n}^{\sigma_n} Z_{B}^{\sigma_1, \sigma_2} X_{B}^{\sigma_1, \sigma_2}
        \Bigg( \prod_{k \neq n} Z_{A}^{\sigma_k} \Bigg)
        \Bigg( \prod_{k} X_{A}^{\sigma_k} \Bigg)
        \Bigg( \prod_{k=1}^{(n-3)/2} X_{A}^{\sigma_{n-2k+1}} Z_{A}^{\sigma_{n-2k+1}} Z_{A}^{\sigma_{n-2k}} X_{A}^{\sigma_{n-2k}} \Bigg) \ket{\Psi} \\
        + \big( X_{A}^{\sigma_n} X_{A}^{\sigma_1} Z_{A}^{\sigma_1} Z_{A}^{\sigma_2} \big)^2 \ket{\Psi}
        \Bigg\rVert
        \leq 3n \sqrt{2\varepsilon_0} + 7 \sqrt{2\varepsilon_1} + 2 \sqrt{2\varepsilon_2} .
    \end{multline}
    Since $n \equiv 3 \pmod{4}$, we can consider successive pairs of terms in the final product of \cref{eq:pair_estimate_step}.
    We can estimate each pair of terms using pair check observables by repeatedly applying the estimate of \cref{lem:product_to_pair_obs} in the first term and commuting the resulting observables of Bob with those of Alice.
    This gives
    \begin{multline}
    \label{eq:paired_step}
        \Bigg\lVert
        X_{B,\sigma_n}^{\sigma_n} Z_{B}^{\sigma_1, \sigma_2} X_{B}^{\sigma_1, \sigma_2}
        \Bigg( \prod_{k=1}^{(n-3)/4} X_{B}^{\sigma_{4k-1}, \sigma_{4k+1}} Z_{B}^{\sigma_{4k-1}, \sigma_{4k+1}} Z_{B}^{\sigma_{4k}, \sigma_{4k+2}} X_{B}^{\sigma_{4k}, \sigma_{4k+2}} \Bigg)
        \Bigg( \prod_{k \neq n} Z_{A}^{\sigma_k} \Bigg) \\
        \Bigg( \prod_{k} X_{A}^{\sigma_k} \Bigg) \ket{\Psi}
        + \big( X_{A}^{\sigma_n} X_{A}^{\sigma_1} Z_{A}^{\sigma_1} Z_{A}^{\sigma_2} \big)^2 \ket{\Psi}
        \Bigg\rVert
        \leq 3n \sqrt{2\varepsilon_0} + \bigg( \frac{9(n-3)}{2} + 7 \bigg) \sqrt{2\varepsilon_1} + (n-1) \sqrt{2\varepsilon_2} .
    \end{multline}

    We may assume the permutation $\sigma$ to in fact be such that all pair check observables appearing in \cref{eq:alice_to_pairs} of \cref{lem:alice_to_pairs} and \cref{eq:paired_step} correspond to the same (pair check round) input for Bob.
    This is compatible with an honest behavior (in which pair check observables correspond to pairs of Pauli observables) since all of these observables $M_{B}^{l,m}$ (where $M$ represents either $X$ or $Z$) have either disjoint or identical indices to all others.
    Specifically, referring to the definition [see \cref{eq:honest_pair_check_observables}] of Bob's observables to be measured upon an input $y$ when $c=2$, we may assume they all correspond to the input $y = \sigma_n$, in which qubit $\sigma_n$ is not to be tested.
    Therefore, after applying the estimate of \cref{lem:alice_to_pairs} to the first term in \cref{eq:paired_step}, we may freely commute all pair check observables and use their involutory property to achieve many cancellations.
    This yields
    \begin{equation}
        \big\lVert X_{A}^{\sigma_n} X_{B,\sigma_n}^{\sigma_n} \ket{\Psi}
        + \big( X_{A}^{\sigma_n} X_{A}^{\sigma_1} Z_{A}^{\sigma_1} Z_{A}^{\sigma_2} \big)^2 \ket{\Psi} \big\rVert
        \leq 3n \sqrt{2\varepsilon_0} + \frac{13(n-1)}{2} \sqrt{2\varepsilon_1} + 2(n-1) \sqrt{2\varepsilon_2} .
    \end{equation}
    Applying the correlations of \cref{eq:local_check_correlations_x} once in the first term and then multiplying on the left by $Z_{A}^{\sigma_2} Z_{A}^{\sigma_1} X_{A}^{\sigma_1} X_{A}^{\sigma_n}$ gives
    \begin{equation}
    \label{eq:3xn_alice_pair_anticomm}
        \big\lVert
        \big\{ X_{A}^{\sigma_1} X_{A}^{\sigma_n}, Z_{A}^{\sigma_1} Z_{A}^{\sigma_2} \big\} \ket{\Psi}
        \big\rVert
        \leq 3n \sqrt{2\varepsilon_0} + \bigg( \frac{13(n-1)}{2} + 1 \bigg) \sqrt{2\varepsilon_1} + 2(n-1) \sqrt{2\varepsilon_2} .
    \end{equation}
    By identical argument to the proof of \cref{prop:anticommutations}, but using \cref{prop:3xn_correlation_estimation,prop:3xn_commutations} instead of \cref{prop:correlation_estimation,prop:commutations} and using the bound of \cref{eq:3xn_alice_pair_anticomm} in place of \cref{lem:alice_pair_anticomm}, this implies
    \begin{equation}
        \big\lVert
        \big\{ X_{A}^{\sigma_1}, Z_{A}^{\sigma_1} \big\} \ket{\Psi}
        \big\rVert
        \leq 3n \sqrt{2\varepsilon_0} + \bigg( \frac{13(n-1)}{2} + 17 \bigg) \sqrt{2\varepsilon_1} + 2(n-1) \sqrt{2\varepsilon_2} .
    \end{equation}
    That $\sigma_1 = i$ yields the result of \cref{eq:3xn_alice_anticomm}.

    To obtain \cref{eq:bob_local_check_anticomm} we use \cref{prop:3xn_correlation_estimation} to write
    \begin{equation}
    \begin{split}
        \big\lVert \big\{ X_{B,i}^{i}, Z_{B,j}^{i} \big\} \ket{\Psi} \big\rVert
        &\leq 4 \sqrt{2\varepsilon_1} + \big\lVert \big\{ X_{A}^{i}, Z_{A}^{i} \big\} \ket{\Psi} \big\rVert \\
        &\leq 3n \sqrt{2\varepsilon_0} + \bigg( \frac{13(n-1)}{2} + 21 \bigg) \sqrt{2\varepsilon_1} + 2(n-1) \sqrt{2\varepsilon_2} ,
    \end{split}
    \end{equation}
    where the final equality follows from \cref{eq:3xn_alice_anticomm} just proved.
\end{proof}

\printbibliography

@article{bell1964einstein,
    title={On the {E}instein {P}odolsky {R}osen paradox},
    author={Bell, John S},
    journal={Phys. Phys. Fiz.},
    volume={1},
    number={3},
    pages={195--200},
    year={1964},
    month={11},
    publisher={American Physical Society},
    doi={10.1103/PhysicsPhysiqueFizika.1.195}
}

@article{mermin1990simple,
    title={Simple unified form for the major no-hidden-variables theorems},
    author={Mermin, N David},
    journal={Phys. Rev. Lett.},
    volume={65},
    number={27},
    pages={3373--3376},
    year={1990},
    month={12},
    publisher={American Physical Society},
    doi={10.1103/PhysRevLett.65.3373}
}

@article{peres1990incompatible,
    title={Incompatible results of quantum measurements},
    author={Peres, Asher},
    journal={Phys. Lett.},
    volume={151},
    number={3},
    pages={107--108},
    year={1990},
    month={12},
    publisher={Elsevier},
    issn={0375-9601},
    doi={10.1016/0375-9601(90)90172-K}
}

@inproceedings{mayers1998quantum,
    author={Mayers, Dominic and Yao, Andrew},
    booktitle={Proceedings 39th Annual Symposium on Foundations of Computer Science (Cat. No. 98CB36280)}, 
    title={Quantum cryptography with imperfect apparatus}, 
    year={1998},
    month={11},
    pages={503--509},
    publisher={IEEE},
    location={Palo Alto, CA, USA},
    issn={0272-5428},
    doi={10.1109/SFCS.1998.743501}
}

@article{mayers2004self,
    author={Mayers, Dominic and Yao, Andrew},
    title={Self testing quantum apparatus},
    year={2004},
    month={7},
    volume={4},
    number={4},
    journal={Quantum Inf. Comput.},
    pages={273--286},
    publisher={Rinton Press, Incorporated},
    issn={1533-7146},
    doi={10.26421/QIC4.4-3}
}

@article{aravind2004quantum,
    title={Quantum mysteries revisited again},
    author={Aravind, Padmanabhan K},
    journal={Am. J. Phys.},
    volume={72},
    number={10},
    pages={1303--1307},
    year={2004},
    month={9},
    publisher={AAPT},
    doi={10.1119/1.1773173}
}

@article{brassard2005quantum,
    title={Quantum pseudo-telepathy},
    author={Brassard, Gilles and Broadbent, Anne and Tapp, Alain},
    journal={Found. Phys.},
    volume={35},
    number={11},
    pages={1877--1907},
    year={2005},
    month={11},
    publisher={Springer},
    issn={1572-9516},
    doi={10.1007/s10701-005-7353-4},
}

@article{navascues2007bounding,
    title={Bounding the set of quantum correlations},
    author={Navascu{\'e}s, Miguel and Pironio, Stefano and Ac{\'i}n, Antonio},
    journal={Phys. Rev. Lett.},
    volume={98},
    number={1},
    pages={010401},
    year={2007},
    month={1},
    publisher={American Physical Society},
    doi={10.1103/PhysRevLett.98.010401}
}

@article{navascues2008convergent,
    title={A convergent hierarchy of semidefinite programs characterizing the set of quantum correlations},
    author={Navascu{\'e}s, Miguel and Pironio, Stefano and Ac{\'i}n, Antonio},
    journal={New J. Phys.},
    volume={10},
    number={7},
    pages={073013},
    year={2008},
    month={7},
    publisher={{IOP} Publishing},
    issn={1367-2630},
    doi={10.1088/1367-2630/10/7/073013}
}

@article{colbeck2011private,
    title={Private randomness expansion with untrusted devices},
    author={Colbeck, Roger and Kent, Adrian},
    journal={J. Phys. A},
    volume={44},
    number={9},
    pages={095305},
    year={2011},
    month={2},
    publisher={{IOP} Publishing},
    issn={1751-8121},
    doi={10.1088/1751-8113/44/9/095305}
}

@article{yang2014robust,
    title={Robust and versatile black-box certification of quantum devices},
    author={Yang, Tzyh Haur and V{\'e}rtesi, Tam{\'a}s and Bancal, Jean-Daniel and Scarani, Valerio and Navascu{\'e}s, Miguel},
    journal={Phys. Rev. Lett.},
    volume={113},
    number={4},
    pages={040401},
    year={2014},
    month={7},
    publisher={American Physical Society},
    doi={10.1103/PhysRevLett.113.040401},
    issn={1079-7114}
}

@article{vazirani2014fully,
    title={Fully device-independent quantum key distribution},
    author={Vazirani, Umesh and Vidick, Thomas},
    journal={Phys. Rev. Lett.},
    volume={113},
    number={14},
    pages={140501},
    year={2014},
    month={9},
    publisher={American Physical Society},
    doi={10.1103/PhysRevLett.113.140501},
    issn={1079-7114}
}

@article{bancal2015physical,
    title={Physical characterization of quantum devices from nonlocal correlations},
    author={Bancal, Jean-Daniel and Navascu\'es, Miguel and Scarani, Valerio and V\'ertesi, Tam\'as and Yang, Tzyh Haur},
    journal={Phys. Rev. A},
    volume={91},
    number={2},
    pages={022115},
    year={2015},
    month={2},
    publisher={American Physical Society},
    doi={10.1103/PhysRevA.91.022115},
    issn={1094-1622}
}

@article{gheorghiu2015robustness,
    title={Robustness and device independence of verifiable blind quantum computing},
    volume={17},
    issn={1367-2630},
    doi={10.1088/1367-2630/17/8/083040},
    number={8},
    journal={New J. Phys.},
    publisher={IOP Publishing},
    author={Gheorghiu, Alexandru and Kashefi, Elham and Wallden, Petros},
    year={2015},
    month={8},
    pages={083040}
}

@misc{hajduvsek2015device,
    title={Device-independent verifiable blind quantum computation}, 
    author={Hajdu{\v{s}}ek, Michal and P{\'e}rez-Delgado, Carlos A and Fitzsimons, Joseph F},
    year={2015},
    month={12},
    eprint={1502.02563},
    archivePrefix={arXiv},
    primaryClass={quant-ph}
}

@article{mckague2016interactive,
    title={Interactive proofs for {BQP} via self-tested graph states},
    volume={12},
    issn={1557-2862},
    doi={10.4086/toc.2016.v012a003},
    number={1},
    journal={Theory Comput.},
    publisher={Theory of Computing},
    author={McKague, Matthew},
    year={2016},
    pages={1--42}
}

@article{mckague2016self,
    title={Self-testing in parallel},
    volume={18},
    number={4},
    journal={New Journal of Physics},
    publisher={IOP Publishing},
    author={Matthew McKague},
    year={2016},
    month={4},
    pages={045013},
    issn={1367-2630},
    doi={10.1088/1367-2630/18/4/045013}
}

@article{wu2016device,
    title={Device-independent parallel self-testing of two singlets},
    author={Wu, Xingyao and Bancal, Jean-Daniel and McKague, Matthew and Scarani, Valerio},
    journal={Phys. Rev. A},
    volume={93},
    number={6},
    pages={062121},
    year={2016},
    month={6},
    publisher={American Physical Society},
    doi={10.1103/PhysRevA.93.062121},
    issn={2469-9934}
}

@misc{coudron2016parallel,
    title={The parallel-repeated magic square game is rigid}, 
    author={Matthew Coudron and Anand Natarajan},
    year={2016},
    month={9},
    eprint={1609.06306},
    archivePrefix={arXiv},
    primaryClass={quant-ph}
}

@article{gheorghiu2017rigidity,
    title={Rigidity of quantum steering and one-sided device-independent verifiable quantum computation},
    volume={19},
    issn={1367-2630},
    doi={10.1088/1367-2630/aa5cff},
    number={2},
    journal={New J. Phys.},
    publisher={IOP Publishing},
    author={Gheorghiu, Alexandru and Wallden, Petros and Kashefi, Elham},
    year={2017},
    month={2},
    pages={023043}
}

@article{kashefi2017garbled,
    author={Kashefi, Elham and Wallden, Petros},
    title={Garbled quantum computation},
    journal={Cryptography},
    volume={1},
    number={1},
    pages={6},
    year={2017},
    month={4},
    issn={2410-387X},
    doi={10.3390/cryptography1010006}
}

@inproceedings{natarajan2017quantum,
    author={Natarajan, Anand and Vidick, Thomas},
    title={A quantum linearity test for robustly verifying entanglement},
    year={2017},
    month={6},
    booktitle={Proceedings of the 49th Annual ACM SIGACT Symposium on Theory of Computing},
    pages={1003--1015},
    publisher={Association for Computing Machinery},
    address={New York, NY, USA},
    series={STOC 2017},
    location={Montreal, Canada},
    doi={10.1145/3055399.3055468}
}

@article{kashefi2017multiparty,
    author={Kashefi, Elham and Pappa, Anna},
    title={Multiparty delegated quantum computing},
    journal={Cryptography},
    volume={1},
    number={2},
    pages={12},
    year={2017},
    month={7},
    issn={2410-387X},
    doi={10.3390/cryptography1020012}
}

@article{fitzsimons2017unconditionally,
    title={Unconditionally verifiable blind quantum computation},
    author={Fitzsimons, Joseph F. and Kashefi, Elham},
    journal={Phys. Rev. A},
    volume={96},
    number={1},
    pages={012303},
    year={2017},
    month={7},
    publisher={American Physical Society},
    doi={10.1103/PhysRevA.96.012303},
    issn={2469-9934}
}

@article{coladangelo2017parallel,
    author={Coladangelo, Andrea},
    title={Parallel self-testing of (tilted) {EPR} pairs via copies of (tilted) {CHSH} and the magic square game},
    year={2017},
    publisher={Rinton Press, Incorporated},
    volume={17},
    number={9–10},
    journal={Quantum Inf. Comput.},
    month={8},
    pages={831--865},
    doi={10.26421/QIC17.9-10-6},
    issn={1533-7146}
}

@article{chao2018test,
    title={Test for a large amount of entanglement, using few measurements},
    author={Chao, Rui and Reichardt, Ben W and Sutherland, Chris and Vidick, Thomas},
    journal={Quantum},
    volume={2},
    pages={92},
    year={2018},
    month={9},
    publisher={Verein zur F{\"o}rderung des Open Access Publizierens in den Quantenwissenschaften},
    doi={10.22331/q-2018-09-03-92},
    issn={2521-327X}
}

@inproceedings{natarajan2018low,
    author={Natarajan, Anand and Vidick, Thomas},
    booktitle={2018 IEEE 59th Annual Symposium on Foundations of Computer Science (FOCS)},
    title={Low-degree testing for quantum states, and a quantum entangled games {PCP} for {QMA}},
    location={Paris, France},
    year={2018},
    month={10},
    pages={731--742},
    publisher={IEEE},
    doi={10.1109/FOCS.2018.00075},
    issn={2575-8454}
}

@article{gheorghiu2019verification,
    title={Verification of quantum computation: An overview of existing approaches},
    volume={63},
    issn={1433-0490},
    doi={10.1007/s00224-018-9872-3},
    number={4},
    journal={Theory Comput. Syst.},
    publisher={Springer},
    author={Gheorghiu, Alexandru and Kapourniotis, Theodoros and Kashefi, Elham},
    year={2019},
    month={5},
    pages={715–-808}
}

@inproceedings{natarajan2019neexp,
    author={Natarajan, Anand and Wright, John},
    booktitle={2019 IEEE 60th Annual Symposium on Foundations of Computer Science (FOCS)},
    title={$\mathsf{NEEXP}$ is contained in $\mathsf{MIP}^{*}$},
    year={2019},
    month={11},
    pages={510--518},
    publisher={IEEE},
    doi={10.1109/FOCS.2019.00039},
    issn={2575-8454}
}

@article{supic2020self,
    title={Self-testing of quantum systems: A review},
    author={{\v{S}}upi{\'c}, Ivan and Bowles, Joseph},
    journal={Quantum},
    volume={4},
    pages={337},
    year={2020},
    month={9},
    publisher={Verein zur F{\"o}rderung des Open Access Publizierens in den Quantenwissenschaften},
    issn={2521-327X},
    doi={10.22331/q-2020-09-30-337}
}

@article{adamson2020quantum,
    title={Quantum magic rectangles: Characterization and application to certified randomness expansion},
    author={Adamson, Sean A. and Wallden, Petros},
    journal={Phys. Rev. Research},
    volume={2},
    number={4},
    pages={043317},
    year={2020},
    month={12},
    publisher={American Physical Society},
    doi={10.1103/PhysRevResearch.2.043317},
    issn={2643-1564}
}

@misc{manvcinska2021constant,
    title={Constant-sized robust self-tests for states and measurements of unbounded dimension}, 
    author={Man{\v{c}}inska, Laura and Prakash, Jitendra and Schafhauser, Christopher},
    year={2021},
    month={3},
    eprint={2103.01729},
    archivePrefix={arXiv},
    primaryClass={quant-ph}
}

@article{supic2021device,
    doi={10.22331/q-2021-03-23-418},
    title={Device-independent certification of tensor products of quantum states using single-copy self-testing protocols},
    author={{\v{S}}upi{\'c}, Ivan and Cavalcanti, Daniel and Bowles, Joseph},
    journal={Quantum},
    issn={2521-327X},
    publisher={Verein zur F{\"{o}}rderung des Open Access Publizierens in den Quantenwissenschaften},
    volume={5},
    pages={418},
    month={3},
    year={2021}
}

@article{sarkar2021self,
    title={Self-testing quantum systems of arbitrary local dimension with minimal number of measurements},
    author={Sarkar, Shubhayan and Saha, Debashis and Kaniewski, J{\k{e}}drzej and Augusiak, Remigiusz},
    journal={npj Quantum Information},
    volume={7},
    number={1},
    pages={151},
    year={2021},
    month={10},
    publisher={Nature Publishing Group},
    issn={2056-6387},
    doi={10.1038/s41534-021-00490-3}
}

@article{ji2021mip,
    title={$\mathsf{MIP}^{*}=\mathsf{RE}$},
    author={Ji, Zhengfeng and Natarajan, Anand and Vidick, Thomas and Wright, John and Yuen, Henry},
    journal={Commun. ACM},
    year={2021},
    month={10},
    pages = {131–-138},
    publisher={Association for Computing Machinery},
    address={New York, NY, USA},
    volume={64},
    number={11},
    issn={0001-0782},
    doi={10.1145/3485628}
}

@article{fu2022constant,
    doi={10.22331/q-2022-01-03-614},
    title={Constant-sized correlations are sufficient to self-test maximally entangled states with unbounded dimension},
    author={Fu, Honghao},
    journal={Quantum},
    issn={2521-327X},
    publisher={Verein zur F{\"{o}}rderung des Open Access Publizierens in den Quantenwissenschaften},
    volume={6},
    pages={614},
    month={1},
    year={2022}
}

\end{document}